\documentclass[fleqn,12pt]{article}

\usepackage{amssymb}
\usepackage{amsmath}
\usepackage{amsthm}
\usepackage{array}
\usepackage{color}
\usepackage[pdftex]{graphicx}

\newcommand{\be}{\begin{equation}}
\newcommand{\ee}{\end{equation}}
\newcommand{\bea}{\begin{eqnarray}}
\newcommand{\eea}{\end{eqnarray}}
\newcommand{\beas}{\begin{eqnarray*}}
\newcommand{\eeas}{\end{eqnarray*}}

\theoremstyle{definition}
\newtheorem{defin}{Definition}[section]
\newtheorem{thm}[defin]{Theorem}
\newtheorem{rem}[defin]{Remark}

\newtheorem{cor}[defin]{Corollary}

\newtheorem{lemma}[defin]{Lemma}
\newtheorem{prop}[defin]{Proposition}

\def\vfi{\varphi}
\def\hil{{\mathcal H}}
\def\kil{{\mathcal K}}

\def\B{{\mathcal B}}

\def\I{\mathcal{I}}
\def\J{\mathcal{J}}
\def\M{\mathcal{M}}
\def\N{\mathcal{N}}

\def\X{{\mathcal X}}

\def\T{\mathcal{T}}
\def\half{\frac{1}{2}}
\def\iff{\Longleftrightarrow}

\def\ep{\varepsilon}
\def\bN{\mathbb{N}}
\def\bC{\mathbb{C}}
\def\bR{\mathbb{R}}

\def\bz{\left(}
\def\jz{\right)}
\def\inv{^{-1}}
\def\kii{\emph}
\def\kiii{}

\newcommand{\ki}{\emph}

\newcommand{\s}{\mbox{ }}
\newcommand{\ds}{\mbox{ }\mbox{ }}
\newcommand{\norm}[1]{\left\| #1\right\|}
\newcommand{\inner}[2]{\langle #1 , #2\rangle}

\newcommand{\vecc}[1]{\underline{#1}}

\newcommand{\sr}[2]{S\bz #1\,||\, #2\jz}
\newcommand{\rsr}[3]{S_{#3}\bz #1\,||\, #2\jz}

\newcommand{\chdist}[2]{C\bz #1\,||\,#2\jz}

\newcommand{\hdist}[3]{H_{#3}\bz #1\,||\,#2\jz}

\newcommand{\psif}[3]{\psi_{#1,#2}\bz #3\jz}
\newcommand{\type}[1]{T_{\vecc{#1}}}

\DeclareMathOperator{\Tr}{Tr}
\DeclareMathOperator{\supp}{supp}

\DeclareMathOperator{\Exp}{\mathbb{E}}
\DeclareMathOperator{\argmax}{argmax}
\DeclareMathOperator{\arcosh}{arcosh}

\begin{document}

\centerline{\LARGE {\bf Quantum state discrimination bounds}}

\centerline{\LARGE {\bf for finite sample size}}
\bigskip
\bigskip

\centerline{\large
Koenraad M.R.~Audenaert$^{1,}$\footnote{E-mail: koenraad.audenaert@rhul.ac.uk},
Mil\'an Mosonyi $^{2,3,}$\footnote{E-mail: milan.mosonyi@gmail.com},
Frank Verstraete$^{4,}$\footnote{E-mail: frank.verstraete@univie.ac.at}
}
\medskip

\begin{center}
$^1$\,Mathematics Department, Royal Holloway, University of London \\
Egham TW20 0EX, United Kingdom
\end{center}
\begin{center}
$^2$\,School of Mathematics,
University of Bristol \\
University Walk, Bristol, BS8 1TW, United Kingdom
\end{center}
\begin{center}
$^3$\,Mathematical Institute, Budapest University of Technology and Economics \\
Egry J\'ozsef u~1., Budapest, 1111 Hungary
\end{center}
\begin{center}
$^4$\,Fakult\"at f\"ur Physik, Universit\"at Wien\\
Boltzmanngasse 5, A-1090 Wien, Austria
\end{center}
\bigskip

\begin{abstract}
\medskip

\noindent In the problem of quantum state discrimination, one has to determine by measurements the state of a quantum system,
based on the a priori side information that the true state is one of two given and completely known states,
$\rho$ or $\sigma$.
In general, it is not possible to decide the identity of the true state with certainty, and
the optimal measurement strategy depends on whether the two possible errors (mistaking $\rho$ for $\sigma$,
or the other way around) are treated as of equal importance or not.
Results on the quantum Chernoff and Hoeffding bounds and the quantum Stein's lemma show that, if several copies of the system are
available then the optimal error probabilities decay exponentially in the number of copies, and the decay
rate is given by a certain statistical distance between $\rho$ and $\sigma$ (the Chernoff distance, the Hoeffding
distances, and the relative entropy, respectively).
While these results provide a complete solution to the asymptotic problem, they are not completely satisfying
from a practical point of view.
Indeed, in realistic scenarios one has access only to finitely many copies of a system, and therefore it is
desirable to have bounds on the error probabilities for finite sample size.
In this paper we provide finite-size bounds on the so-called
Stein errors, the Chernoff errors, the Hoeffding errors and the mixed error probabilities related to the Chernoff and the Hoeffding errors.

\medskip

\end{abstract}

\bigskip

\section{Introduction}

Assume we have a quantum system with a finite-dimensional Hilbert space $\hil$, and we know
that the system has been prepared either in state $\rho$ (this is our \ki{null hypothesis}
$H_0$) or state $\sigma$ (this is our \ki{alternative hypthesis} $H_1$). (By a state we
mean a density operator, i.e., a positive semi-definite operator with trace $1$).
The goal of state discrimination is to come up with a ``good'' guess for the true state,
based on measurements on the system. By ``good'' we mean that some error probability is
minimal; we will specify this later. We will study the asymptotic scenario, where we assume
that several identical and independent (i.i.d.) copies of the system are available, and we are
allowed to make arbitrary collective measurements on the system. Due to the
i.i.d.~assumption, i.e.,  that
the copies are identical and independent, the joint state of the $n$-copy system is
either $\rho_n:=\rho^{\otimes n}$ or $\sigma_n:=\sigma^{\otimes n}$ for every $n\in\bN$.

A \ki{test} on $n$ copies is an operator $T\in\B(\hil^{\otimes n}),\,0\le T\le I_n$, that determines the binary POVM
$(T,I_n-T)$. If the outcome corresponding to $T$ occurs then we accept the null hypothesis to be true,
otherwise we accept the alternative hypothesis. Of course, we might make an error by concluding that the
true state is $\sigma$ when it is actually $\rho$ (\ki{error of the first kind} or \ki{type I error}) or the other way around
(\ki{error of the second kind} or \ki{type II error}). The probabilities of these errors when the measurement $(T,I_n-T)$ was
performed are given by
\begin{equation*}
\alpha_n(T):=\Tr\rho_n(I_n-T)\ds\ds\text{(first kind)\ds\ds and}\ds\ds
\beta_n(T):=\Tr\sigma_nT\ds\ds\text{(second kind)}.
\end{equation*}
Unless $\rho_n$ and $\sigma_n$ are perfectly distinguishable (which is the case if and only if $\supp\rho_n\perp\supp\sigma_n$), the
two error probabilities cannot be simultaneously eliminated, i.e., $\alpha_n(T)+\beta_n(T)>0$ for any test
$T$, and our aim is to find a joint optimum of the two error probabilities, according to some criteria.

In a Bayesian approach, one considers the scenario where $\rho$ and $\sigma$ are prepared
with some prior probabilities $p$ and $1-p$, respectively; the natural quantities to consider
in this case are the so-called \ki{Chernoff errors}, given by
$\min_{T \text{ test}}\{p\alpha_n(T)+(1-p)\beta_n(T)\}$.
More generally, consider for any $\kappa,\lambda>0$ the quantities
\begin{equation*}
e_{n,\kappa,\lambda}:=\min_{T\text{ test}}\{\kappa\alpha_n(T)+\lambda\beta_n(T)\}.
\end{equation*}
For a self-adjoint operator $X$ and constant $c\in\bR$, let $\{X>c\}$ denote the spectral projection of $X$ corresponding to the interval $(c,+\infty)$. We define
$\{X\ge c\}$, $\{X<c\}$ and $\{X\le c\}$ similarly.
As one can easily see,
\begin{equation*}
e_{n,\kappa,\lambda}=\frac{\kappa+\lambda}{2}-\half\norm{\kappa\rho_n-\lambda\sigma_n}_1,
\end{equation*}
(where $\norm{X}_1:=\Tr|X|$ for any operator $X$),
and the minimum is reached at any test $T$ satisfying
\begin{equation*}
\{\kappa\rho_n-\lambda\sigma_n>0\}\le T\le\{\kappa\rho_n-\lambda\sigma_n\ge 0\}.
\end{equation*}
Such a test is called a \ki{Neyman-Pearson test} or \ki{Holevo-Helstr\"om test} in the literature \cite{Helstrom,Holevo}. By the above, such tests are optimal from the point of view of trade-off between the two error probabilities. Indeed, if $T$ is a Neyman-Pearson test corresponding to $\kappa$ and $\lambda$ then for any other test $S$ we have
\begin{equation*}
\kappa\alpha_n(T)+\lambda\beta_n(T)\le \kappa\alpha_n(S)+\lambda\beta_n(S).
\end{equation*}
In particular, if $\alpha_n(S)<\alpha_n(T)$ then necessarily $\beta_n(S)>\beta_n(T)$ and vice versa, i.e., if $S$ performs better than a Neymann-Pearson test for one of the error probabilities then it necessarily performs worse for the other. This is the so-called
\ki{quantum Neyman-Pearson lemma}.
For later use, we introduce the notations
\begin{equation}\label{NP}
\N_{n,a}:=\{T\text{ test}\,:\,\{e^{-na}\rho_n-\sigma_n>0\}\le T\le\{e^{-na}\rho_n-\sigma_n\ge 0\}
\}
\end{equation}
and
\begin{align}
e_n(a)&:=e_{n,e^{-na},1}=\min_{T\text{ test}}\{e^{-na}\alpha_n(T)+\beta_n(T)\}\nonumber\\
&=e^{-na}\alpha_n(T)+\beta_n(T),\ds T\in\N_{n,a},\label{mixed error def}
\end{align}
where $a\in\bR$ is a parameter.

The following has been shown for the i.i.d.~case in \cite{Aud,NSz} (see also
\cite{HMO,HMO2,HMH,MHOF,M} for various generalizations to correlated settings).
\begin{thm}\label{thm:Chernoff}
For any $\kappa,\lambda>0$ we have
\begin{align*}
-\lim_{n\to\infty}\frac{1}{n}\log e_{n,\kappa,\lambda}=
-\lim_{n\to\infty}\frac{1}{n}\log e_n(0)=\chdist{\rho}{\sigma}:=-\inf_{0\le t\le 1}
\log\Tr\rho^t\sigma^{1-t},
\end{align*}
where $\chdist{\rho}{\sigma}$ is called the \ki{Chernoff distance} of $\rho$ and $\sigma$.
\end{thm}

Another natural way to optimize the two error probabilities is to put a constraint on one of them and optimize the other one under this constraint.
It is usual to optimize the type II error under the constraint that the type I error is kept under a constant error bar $\ep\in(0,1)$, in which case we are interested in the quantities
\begin{equation}\label{beta ep}
\beta_{n,\ep}:=\min\{\beta_n(T)\,:\,T\text{ test, }\alpha_n(T)\le\ep\}.
\end{equation}
Another natural choice is when an exponential constraint is imposed on the type I error, which gives
\begin{equation}\label{beta r}
\beta_{n,e^{-nr}}:=\min\{\beta_n(T)\,:\,T\text{ test, }\alpha_n(T)\le e^{-nr}\}
\end{equation}
for some fixed parameter $r>0$. Unlike for the quantities $e_{n,\kappa,\lambda}$ above,
there are no explicit expressions known for the values of $\beta_{n,\ep}$ and $\beta_{n,e^{-nr}}$, or for the tests achieving them. However, the asymptotic behaviours are known also in these cases.
The asymptotics of $\beta_{n,\ep}$ is given by the quantum Stein's lemma, first proved for the i.i.d.~case in \cite{HP,ON} and later generalized to various correlated scenarios in
\cite{BS-Sch,BP,HP2,HMO2,HMH,MHOF,M}.
\begin{thm}\label{thm:Stein}
We have
\begin{equation*}
-\lim_{n\to+\infty}\frac{1}{n}\log\beta_{n,\ep}=
\inf_{\{T_n\}_{n\in\bN}}
\left\{-\lim_{n\to\infty}\frac{1}{n}\log\beta_n(T_n)\,:\,\lim_{n\to\infty}\alpha_n(T_n)=0 \right\}
=
\sr{\rho}{\sigma},
\end{equation*}
where the infimimum is taken over all sequences of measurements for which the indicated limit exists, and $\sr{\rho}{\sigma}$ is the relative entropy of $\rho$ with respect to $\sigma$.
\end{thm}
The asymptotics of $\beta_{n,e^{-nr}}$ has been an open problem for a long time (see, e.g., \cite{HO}), which was finally solved for the i.i.d.~case
in \cite{Hayashi} and \cite{Nagaoka} (apart from some minor technicalities that
were treated both in \cite{ANSzV} and \cite{HMO2}), based on the techniques developed in
\cite{Aud} and \cite{NSz}.
These results were later extended to various correlated settings in \cite{HMO2,HMH,MHOF,M}.
\begin{thm}\label{thm:Hoeffding}
For any $r>0$ we have
\begin{equation*}
-\lim_{n\to\infty}\frac{1}{n}\log\beta_{n,e^{-nr}}=
\hdist{\rho}{\sigma}{r}:=
-\inf_{0\le t<1}\left\{\frac{t r+\log\Tr\rho^{t}\sigma^{1-t}}{1-t}\right\},
\end{equation*}
where $\hdist{\rho}{\sigma}{r}$ is the \ki{Hoeffding distance} of $\rho$ and $\sigma$ with parameter $r$.
\end{thm}

It is not too difficult to see that Theorem \ref{thm:Hoeffding} can also be reformulated in the following way:
\begin{align*}
\hdist{\rho}{\sigma}{r}=\inf_{\{T_n\}_{n\in\bN}}
\left\{-\lim_{n\to\infty}\frac{1}{n}\log\beta_n(T_n)\,:\,\limsup_{n\to\infty}\frac{1}{n}\log\alpha_n(T_n)\le -r\right\},
\end{align*}
where the infimum is taken over all possible sequences of tests for which the indicated limit exists (see \cite{HMO2} for details). This formulation makes it clear that the Hoeffding distance quantifies the trade-off between the two error probabilities in the sense that it gives the optimal exponential decay of the error of the second kind under the constraint that the error of the first kind decays with a given exponential speed.

While there is no explicit expression known for the optimal tests minimizing
\eqref{beta ep} and \eqref{beta r}, it is known that the Neyman-Pearson tests are asymptotically optimal for this problem in the sense given in Theorem \ref{thm:NP asymptotics} below.
For a positive semidefinite operator $X$ and $x\in\bR$, let $P_x$ denote the spectral projection of $X$ corresponding to the singleton $\{x\}$. For every $t\in\bR$, we define
$X^t:=\sum_{x>0}x^tP_x$; in particular,
$X^0$ denotes the projection onto the support of $X$, i.e., $X^0=\{X>0\}$.
The following was given in \cite{HMO2}:
\begin{thm}\label{thm:NP asymptotics}
For any $r>-\log\Tr\rho\sigma^0$, let
$a_r:=\hdist{\rho}{\sigma}{r}-r$.
For any sequence of tests $\{T_n\}$ satisfying $T_n\in\N_{n,a_r},\,n\in\bN$, we have
\begin{align*}
-\lim_{n\to\infty}\frac{1}{n}\log\alpha_n(T_n)&=\hat\vfi(a_r)=r,\\
-\lim_{n\to\infty}\frac{1}{n}\log\beta_n(T_n)&=
-\lim_{n\to\infty}\frac{1}{n}\log e_n(a_r)=\vfi(a_r)=
\hdist{\rho}{\sigma}{r},
\end{align*}
where for every $a\in\bR$,
\begin{equation}\label{phi}
\vfi(a):=\max_{0\le t\le 1}\{at-\log\Tr\rho^t\sigma^{1-t}\},\ds\ds\ds
\hat\vfi(a):=
\vfi(a)-a.
\end{equation}
\end{thm}

Theorems \ref{thm:Chernoff}--\ref{thm:NP asymptotics} give a complete solution to the asymptotic problem
in the most generally considered setups. These results, however, rely on the assumption that one has access to an unlimited number of identical copies of the system in consideration, which of course is never satisfied in reality. Note also that the above results give no information about the error probabilities for finite sample size, which is the relevant question from a practical point of view.
Our aim in this paper is to provide bounds on the finite-size error probabilities that can be more useful for applications.
There are two similar but slightly different ways to do so;
one is to consider the
optimal type II errors for finite $n$; we treat this in Section \ref{sec:typeII}.
The other is to study the
asymptotic behaviour of the error probabilities corresponding to the
Holevo-Helstr\"om measurements, that are known to be asymptotically optimal; we
provide bounds on these error probabilities in Section \ref{sec:mixed}.
In the special case where both hypotheses are classical binary probability measures, a
direct computation yields bounds on the mixed error  probabilities $e_n(a)$; we present this
in the Appendix.
Some of the technical background is summarized in Section
\ref{sec:preliminaries} below.

\section{Preliminaries}\label{sec:preliminaries}

\subsection{R\'enyi relative entropies and related measures}

For positive semidefinite operators $A,B$ on a Hilbert space $\kil$,
we define their
\ki{R\'enyi relative entropy} with parameter $t\in [0,+\infty)\setminus\{1\}$ as
\begin{equation*}
\rsr{A}{B}{t}:=\begin{cases}
\frac{1}{t-1}\log\Tr A^tB^{1-t}=\frac{1}{t-1}\psif{A}{B}{t},& \text{if}\ds t\in [0,1)\s\text{or}\s \supp A\le\supp B,\\
+\infty,& \text{otherwise},
\end{cases}
\end{equation*}
where
\begin{equation*}
\psif{A}{B}{t}:=\log Z_{A,B}(t),\ds\ds\ds Z_{A,B}(t):=\Tr A^t B^{1-t},\ds\ds\ds t\in\bR.
\end{equation*}
Here we use the convention $\log 0:=-\infty$ and $0^t:=0$, i.e., all powers are computed on
the
supports of $A$ and $B$, respectively. In particular, $\rsr{A}{B}{t}=+\infty$ if and only
if
$\supp A\perp\supp B$ and $t\in[0,1)$, or $\supp A\nleq\supp B$ and $t>1$.
Note that $Z_{A,B}(t)$ is a quasi-entropy in the sense of \cite{Petz}.
For most of what follows, we fix $A$ and $B$, and hence we omit them from the subscripts, i.e., we use $\psi$ instead of $\psi_{A,B}$, etc.

If $p$ is a positive measure on some finite set $\X$ then it can be naturally identified
with a positive function, which we will denote the same way, i.e., we
have the identity $p(\{x\})=p(x),\,x\in\X$. Moreover, $p$ can be naturally identified with a
positive semidefinite operator on $\bC^{\X}=l^2(\X)$, which
we again denote the same way; the matrix of this operator is given by
$\inner{e_x}{p e_y}=\delta_{x,y}p(x)$, where
$\{e_x\}_{x\in\X}$ is the canonical basis of $\bC^{\X}$. Given this identification, we
can use the above definition to define the R\'enyi relative entropies of positive measures/functions
$p$ and $q$ on some finite set $\X$, and we get
$\rsr{p}{q}{t}=\frac{1}{t-1}\log\sum_{x\in\X}p(x)^t q(x)^{1-t}$ whenever $\rsr{p}{q}{t}$ is finite.

Let $A=\sum_{i\in\I} a_i P_i$ and $B=\sum_{j\in\J} b_j Q_j$ be decompositions of the positive semidefinite operators $A$ and $B$ such that
$\{P_i\}$ and $\{Q_j\}$ are sets of orthogonal projections and $a_i,b_j>0$ for all $i$ and $j$.
Let $\X_{A,B}:=\{(i,j)\,:\,\Tr P_i Q_j>0\}$,
and define
\begin{equation}\label{pq def}
p_{A,B}(i,j):=a_i\Tr P_iQ_j,\ds\ds\ds
q_{A,B}(i,j):=b_j\Tr P_iQ_j,\ds\ds\ds
(i,j)\in\X_{A,B}.
\end{equation}
Then $p=p_{A,B}$ and $q=q_{A,B}$ are positive measures on $\X=\X_{A,B}$, and we have
\begin{equation*}
\psi_{A,B}(t)=\psi_{p,q}(t),\ds t\in\bR,\ds\ds\text{and}\ds\ds
\rsr{A}{B}{t}=\rsr{p}{q}{t},\ds t\in[0,1).
\end{equation*}
It is easy to see that
\begin{equation*}
\supp p=\supp q=\X\ds\ds\ds\ds\text{and}\ds\ds\ds\ds
p(\X)=\Tr AB^0,\ds q(\X)=\Tr A^0 B.
\end{equation*}

Note that the decompositions of $A$ and $B$ are not unique, and hence neither are the set
$\X$ and the measures $p$ and $q$. However, if $\X,p,q$ are defined through some
decompositions
$A=\sum_{i\in\I} a_i P_i$ and $B=\sum_{j\in\J} b_j Q_j$ then we will always assume that for
every $n\in\bN$,
$\X_{A^{\otimes n},B^{\otimes n}}, p_{A^{\otimes n},B^{\otimes n}}$ and $q_{A^{\otimes n},B^{\otimes n}}$
are defined through the decompositions $A^{\otimes n}=\sum_{\vecc{i}\in\I^n} a_{\vecc{i}}
P_{\vecc{i}}$ and $B=\sum_{\vecc{j}\in\J^n} b_{\vecc{j}} Q_{\vecc{j}}$, where
$a_{\vecc{i}}:=a_{i_1}\cdot\ldots\cdot a_{i_n}$,
$P_{\vecc{i}}:=P_{i_1}\otimes\ldots\otimes P_{i_n}$, etc. In this way, we have
\begin{equation*}
\X_{A^{\otimes n},B^{\otimes n}}=\X_{A,B}^n,\ds\ds\ds
p_{A^{\otimes n},B^{\otimes n}}=p_{A,B}^{\otimes n},\ds\ds\ds
q_{A^{\otimes n},B^{\otimes n}}=q_{A,B}^{\otimes n}.
\end{equation*}
The above mapping of pairs of positive semi-definite operators to pairs of classical
positive measures was used to prove the optimality of the quantum Chernoff bound in
\cite{NSz} and subsequently the optimality of the quantum
Hoeffding bound in \cite{Nagaoka}, by mapping the quantum state discrimination problem
into a classical one. We will use the same approach to give lower bounds on the mixed error
probabilities in Section \ref{sec:mixed}.

For given $A,B$, and every $t\in\bR$, define a probability measure $\mu^t$ on $\X$ as
\begin{equation*}
\mu^t(i,j):=\frac{1}{Z(t)}p(i,j)^tq(i,j)^{1-t},\ds\ds\ds (i,j)\in\X,
\end{equation*}
where $\X,p,q$ are given as above, and $Z(t)=Z_{A,B}(t)=\sum_{i,j}p(i,j)^tq(i,j)^{1-t},\,t\in\bR$.

\begin{lemma}\label{lemma:convexity}
The function $\psi$ is convex on $\bR$, it is affine if and only if $q$ is a constant multiple of $p$ and otherwise $\psi''(t)>0$ for all $t\in\bR$.
\end{lemma}
\begin{proof}
A straightforward computation shows that
\begin{align}
\psi'(t)&=Z(t)\inv\sum_{i,j} p(i,j)^tq(i,j)^{1-t}(\log p(i,j)-\log
q(i,j))=\Exp_{\mu^t}f,\label{psi derivative}\\
\psi''(t)&=
\Exp_{\mu^t}(f)^2-(\Exp_{\mu^t}f)^2\nonumber\\
&=\frac{\Tr A^t B^{1-t}(\log A-\log B)^2}{\Tr A^t B^{1-t}}
-
\bz\frac{\Tr A^t B^{1-t}(\log A-\log B)}{\Tr A^t B^{1-t}}\jz^2
,\label{second derivative=variation}
\end{align}
where $f(i,j):=\log p(i,j)-\log q(i,j),\,(i,j)\in\X$, and $\Exp_{\mu^t}$ denotes the expectation value
with respect to $\mu^t$. This shows that $\psi$ is convex on the whole real line, and
$\psi''(t)=0$ for some $t\in\bR$ if and only if $f$ is constant, which is equivalent to $q$
being a constant multiple of $p$. Since this condition for a flat second derivative is
independent of $t$, the assertion follows.
\end{proof}
For a condition for a flat derivative of $\psi$ in terms of $A$ and $B$, see Lemma 3.2 in
\cite{HMO2}.
\begin{cor}\label{cor:monotonicity}
If $\Tr A\le 1$ then the function $t\mapsto\rsr{A}{B}{t}$ is monotone increasing on $[0,1)$ and on $(1,+\infty)$. If $\Tr A=1$ then we have
\begin{equation*}
\lim_{t\to 1}\rsr{A}{B}{t}=\rsr{A}{B}{1}:=\sr{A}{B}:=
\begin{cases}
\Tr A(\log^* A-\log^* B),&\supp A\le\supp B,\\
+\infty,&\text{otherwise},
\end{cases}
\end{equation*}
where $\log^* x=\log x,\,x>0$, and $\log^*0:=0$.
\end{cor}
\begin{proof}
We have $\frac{d}{dt}\rsr{A}{B}{t}=\frac{\psi'(t)(t-1)-\psi(t)}{(t-1)^2}=
\frac{-\psi(1)}{(t-1)^2}+\half\psi''(\xi_t)$, where $\xi_t$ is between $t$ and $1$. The first
assertion then follows due to Lemma \ref{lemma:convexity}. If $\Tr A=1$ then
$\lim_{t\to 1}\rsr{A}{B}{t}=\psi'(1)$, which is easily seen to be equal to $\sr{A}{B}$.
\end{proof}

The quantity $\sr{A}{B}$ defined above is the \ki{relative entropy} of $A$ with respect to $B$. The following Lemma complements Corollary \ref{cor:monotonicity}:

\begin{lemma}\label{lemma:Renyi bounds}
Assume that $\supp A\le\supp B$.
For any $c>0$ and $|t-1|\le \delta:=\min\left\{\half,\frac{c}{2\log\eta}\right\}$, where
\begin{equation}\label{eta def}
\eta:=1+e^{\half\rsr{A}{B}{3/2}}+e^{-\half\rsr{A}{B}{1/2}},
\end{equation}
we have
\begin{align}
\rsr{A}{B}{t}
&\ge
\sr{A}{B}-(4\cosh c)(1-t)(\log \eta)^2,\ds\ds t\in(1-\delta,1),\label{TCR1}\\
\rsr{A}{B}{t}
&\le
\sr{A}{B}+(4\cosh c)(t-1)(\log \eta)^2,\ds\ds t\in(1,1+\delta).\label{TCR2}
\end{align}
With the convention $\rsr{A}{B}{1}:=\sr{A}{B}$, the above inequalities can be combined into
\begin{align*}
\rsr{A}{B}{\beta}
&\le
\rsr{A}{B}{t}+(4\cosh c)(\log \eta)^2(\beta-t),\ds\ds
1-\delta\le t\le 1\le\beta\le 1+\delta.
\end{align*}
\end{lemma}
\begin{proof}
The inequality \eqref{TCR2} was first given for conditional entropies in \cite{TCR} and
for relative entropies of states in \cite{Tomamichel}.
Exactly the same proof yields \eqref{TCR2} for general positive semidefinite operators, and
also the inequality \eqref{TCR1}.
\end{proof}

For an operator $X$ on a finite-dimensional Hilbert space, let
$\norm{X}_1:=\Tr|X|=\Tr\sqrt{X^*X}$ denote its trace-norm.
The \ki{von Neumann entropy} of a positive semi-definite operator $A$ is defined as
$S(A):=-\Tr A\log A=-\sr{A}{I}$. The following is a sharpening of the Fannes inequality
\cite{Fannes}; for a proof, see, e.g., \cite{Aud2} or \cite{Petzbook}.
\begin{lemma}\label{lemma:Fi}
For density operators $A$ and $B$ on a finite-dimensional Hilbert space $\hil$,
\begin{equation*}
|S(A)-S(B)|\le\half\norm{A-B}_1\log(\dim\hil-1)+h_2(\norm{A-B}_1/2),
\end{equation*}
where $h_2(x):=-x\log x-(1-x)\log(1-x),\,x\in[0,1]$.
\end{lemma}
\medskip

For positive semidefinite operators $A$ and $B$, we define their \ki{Chernoff distance} as
\begin{equation*}
\chdist{A}{B}:=-\min_{0\le t\le 1}\log\Tr A^tB^{1-t}=\sup_{0\le t<1}\left\{(1-t)\rsr{A}{B}{t}\right\}.
\end{equation*}
The following inequality between the trace-norm and the Chernoff distance was given in
Theorem 1 of \cite{Aud}; see also the simplified proof by N.~Ozawa in \cite{JOPS}.
\begin{lemma}\label{lemma:Aud}
Let $A$ and $B$ be positive semidefinite operators on a finite-dimensional Hilbert space $\hil$. Then
\begin{equation*}
\half\Tr(A+B)-\half\Tr|A-B|\le\Tr A^tB^{1-t},\ds\ds\ds t\in [0,1],
\end{equation*}
or equivalently,
\begin{equation*}
\half\Tr(A+B)-\half\norm{A-B}_1\le e^{-\chdist{A}{B}}.
\end{equation*}
\end{lemma}
The above lemma was used to prove the achievability of the quantum Chernoff bound in
\cite{Aud}, and subsequently the achievability of the quantum Hoeffding bound in
\cite{Hayashi}. We will recall these results in Section \ref{sec:mixed}.
\bigskip

The \ki{Hoeffding distance} of $A$ and $B$ with parameter $r\ge 0$ is defined as
\begin{equation*}
\hdist{A}{B}{r}:=
\sup_{0\le t<1}\left\{\rsr{A}{B}{t}-\frac{tr}{1-t}\right\}=
\sup_{0\le t<1}\frac{-tr-\log\Tr A^tB^{1-t}}{1-t}
\end{equation*}
(cf.~Theorem \ref{thm:Hoeffding} for the same expression for density operators).
For every $a\in\bR$, let
\begin{align*}
\vfi(a):=\max_{t\in[0,1]}\{ta-\psi(t)\},\ds\ds\ds
\hat\vfi(a):=\max_{t\in[0,1]}\{(t-1)a-\psi(t)\}=\vfi(a)-a.
\end{align*}
as in \eqref{phi}.

\begin{lemma}\label{lemma:Hoeffding representation}
\begin{enumerate}
\item
The function $r\mapsto\hdist{A}{B}{r}$ is convex and monotonic decreasing.
\item\label{Hoeffding limit}
$\lim_{r\searrow 0}\hdist{A}{B}{r}=\hdist{A}{B}{0}$, and if
$\Tr A=1$ then $\hdist{A}{B}{0}=\sr{A}{B}$.
\item\label{Hoeffding representation}
For every $-\psi(1)<r<-\psi(0)-\psi'(0)$ there exists a unique $t_r\in(0,1)$ such that
\begin{equation*}
r=\sr{\mu^{t_r}}{p}=(t_r-1)\psi'(t_r)-\psi(t_r),\ds\ds\ds
\hdist{A}{B}{r}=\sr{\mu^{t_r}}{q}=t_r\psi'(t_r)-\psi(t_r).
\end{equation*}
\item\label{Hoeffding representation2}
For every $r>-\psi(1)$ there is a unique $a_r\in\bR$ such that
\begin{equation}\label{phi ar}
\vfi(a_r)=\hdist{A}{B}{r},\ds\ds\ds\ds\ds
\hat\vfi(a_r)=r.
\end{equation}
Moreover, $a_r=\hdist{A}{B}{r}-r$, and if $r<-\psi(0)-\psi'(0)$ then
$a_r=\psi'(t_r)$ with the $t_r$ given in \ref{Hoeffding representation}.

\end{enumerate}
\end{lemma}
\begin{proof}
The first assertion is obvious from the definition, and the second identity in \ref{Hoeffding limit} follows immediately from Corollary \ref{cor:monotonicity}. Note that
\begin{equation*}
\hdist{A}{B}{r}=
\sup_{0\le t<1}\frac{-tr-\psi(t)}{1-t}=\sup_{s\ge 0}\{-sr-\tilde\psi(s)\},
\end{equation*}
where $\tilde\psi(s):=(1+s)\psi\bz\frac{s}{1+s}\jz$, and hence the function
$r\mapsto\hdist{A}{B}{r}$ is essentially the Legendre transform of $\tilde\psi$.
By Proposition 4.1 and Corollary 4.1 in \cite{ET}, $\tilde\psi^*$ is lower semicontinuous,
and hence $\liminf_{r\searrow 0}\hdist{A}{B}{r}\ge\hdist{A}{B}{0}\ge \lim_{r\searrow 0}\hdist{A}{B}{r}$, where the second inequality is due to the monotonicity in $r$. This gives the first identity in \ref{Hoeffding limit}.

Convexity of $\psi$ yields that $\psi(0)+\psi'(0)\le\psi(1)$ and equality holds if and only if $\psi$ is affine, in which case the assertion in \ref{Hoeffding representation} is empty and hence for the rest we assume $\psi''(t)>0,\,t\in\bR$.
By the definition of $\tilde\psi$,
\begin{equation*}
\tilde\psi'(s)=\psi\bz\frac{s}{1+s}\jz+\frac{1}{1+s}\psi'\bz\frac{s}{1+s}\jz\ds\ds\text{and}\ds\ds
\tilde\psi''(s)=\frac{1}{(1+s)^3}\psi''\bz\frac{s}{1+s}\jz,
\end{equation*}
and hence $\tilde\psi$ is also convex.
Note that
$\tilde\psi'(0)=\psi(0)+\psi'(0)$ and $\lim_{s\to+\infty}\tilde\psi'(s)=\psi(1)$, and hence,
\begin{equation*}
\hdist{A}{B}{r}=\sup_{s\ge 0}\{-sr-\tilde\psi(s)\}=
\begin{cases}
-\tilde\psi(0)=-\psi(0),& -r\le\psi(0)+\psi'(0),\\
+\infty &  -r>\psi(1).
\end{cases}
\end{equation*}
On the other hand, for any $-\psi(1)<r<-\psi(0)-\psi'(0)$ there exists a unique $s_r>0$ such that \begin{equation*}
-r=\tilde\psi'(s_r)
=\psi(t_r)+(1-t_r)\psi'(t_r)\ds\ds\text{and}\ds\ds
\hdist{A}{B}{r}=
-s_rr-\tilde\psi(s_r)
 =-\psi(t_r)+t_r\psi'(t_r),
\end{equation*}
where $t_r=\frac{s_r}{1+s_r}\in (0,1)$.
The identities
\begin{equation*}
\sr{\mu^{t}}{p}=(t-1)\psi'(t)-\psi(t),\ds\ds\ds
\sr{\mu^{t}}{q}=t\psi'(t)-\psi(t),\ds\ds\ds t\in\bR,
\end{equation*}
follow by a straightforward computation. This proves \ref{Hoeffding representation}.
For $-\psi(1)<r<-\psi(0)-\psi'(0)$, \ref{Hoeffding representation2} is an immediate
consequence of \ref{Hoeffding representation}. For the general case, see e.g., Theorem 4.8
in \cite{HMO2}.
\end{proof}

\begin{rem}
The equation of the tangent line of $\psi$ at point $t$ is $l(x):=\psi(t)+(x-t)\psi'(t)$. Hence, $\psi(t)-t\psi'(t)=-\sr{\mu^{t}}{q}$ is its intersection with the $y$ axis and
$\psi(t)-(t-1)\psi'(t)=-\sr{\mu^{t}}{p}$ is its intersection with the $x=1$ line.
\end{rem}

\begin{rem}
Note that
\begin{align*}
\psi(0)&=\log\Tr A^0B=\log q(\X),\ds
\text{and if}\ds A^0\ge B^0\ds\text{then}\ds
\psi'(0)=-\sr{B}{A}/\Tr B,\\
\psi(1)&=\log\Tr AB^0=\log p(\X),\ds
\text{and if}\ds A^0\le B^0\ds\text{then}\ds
\psi'(1)=\sr{A}{B}/\Tr A.
\end{align*}
\end{rem}

\begin{rem}
It was shown in \cite{Hoeffding} that
\begin{equation*}
\hdist{p}{q}{r}=\inf\{\sr{\mu}{q}\,:\,\sr{\mu}{p}\le r\},
\end{equation*}
where $p$ and $q$ are probability distributions on some finite set $\X$, and $\mu^{t_r}$ with the $t_r$ given in Lemma \ref{lemma:Hoeffding representation} is a unique minimizer in the above expression. However, the above representation of the Hoeffding distance does not hold in the quantum case. Indeed, it was shown in \cite{HO,ON} that for density operators $\rho$ and $\sigma$,
\begin{align*}
\inf\{\sr{\tilde\rho}{\sigma}\,:\,\tilde\rho\text{ is a density operator, }\sr{\tilde\rho}{\rho}\le r\}
&=
\sup_{0\le t<1}\frac{-tr-\log\Tr e^{t\rho+(1-t)\sigma}}{1-t}\\
&\ge \hdist{\rho}{\sigma}{r},
\end{align*}
where the inequality is due to the Golden-Thompson inequality (see, e.g., Theorem IX.3.7 in \cite{Bhatia}), and is in general strict.
\end{rem}

Although the Chernoff distance and the Hoeffding distances don't satisfy the axioms of
a metric on the set of density operators
(the Chernoff distance is symmetric but does not satisfy the triangle inequality, while the
Hoeffding distances are not even symmetric), the Lemma below gives some motivation
why they are called ``distances''.

\begin{lemma}
If $\Tr A\le 1$ and $\Tr B\le 1$ then
\begin{equation*}
\rsr{A}{B}{t}\ge 0,\ds\ds\ds
\chdist{A}{B}\ge 0,\ds\ds\ds
\hdist{A}{B}{r}\ge 0
\end{equation*}
for every $t\in(0,+\infty)\setminus\{1\}$ and every $r\ge 0$. Moreover, the above inequalities are strict unless $A=B$ and $\Tr A=1$ or $r>-\psi(0)-\psi'(0)$.
\end{lemma}
\begin{proof}
H\"older's inequality (see Corollary IV.2.6 in \cite{Bhatia}) yields that $\Tr A^tB^{1-t}\le(\Tr A)^t(\Tr B)^{1-t}$ for every
$t\in[0,1]$, from which the assertions follow easily, taking into account the previous Lemmas.
\end{proof}

\subsection{Types}

Let $\X$ be a finite set and let $\M(\X)$ denote the set of non-zero positive measures on
$\X$ and $\M_1(\X)$ the set of probability measures on $\X$.
We will identify positive measures with positive semidefinite operators as described in the
previous subsection.
For $\mu\in\M(\X)$ let $S(\mu):=-\sum_{x\in\X}\mu(x)\log\mu(x)$ be its entropy,
and for $\mu_1,\mu_2\in\M(\X)$
let the relative entropy of $\mu_1$ and $\mu_2$ be defined as
$\sr{\mu_1}{\mu_2}:=\sum_{x\in\X}\mu_1(x)\log\frac{\mu_1(x)}{\mu_2(x)}$
if $\supp\mu_1\le\supp\mu_2$, and $+\infty$ otherwise.

For a sequence $\vecc{x}\in\X^n$, the \ki{type} of $\vecc{x}$ is the probability
distribution given by
\begin{equation*}
\type{x}(y):=\frac{1}{n}|\{k\,:x_k=y\}|,\ds\ds\ds y\in\X,
\end{equation*}
where $|H|$ denotes the cardinality of a set $H$.
Note that $\type{x}=\type{y}$ if and only if $\vecc{x}$ is a permutation of $\vecc{y}$.
Obviously, if $\mu\in\M(\X)$ then the measure of an $\vecc{x}\in\X^n$ with respect to $\mu^{\otimes n}$ only depends on the type of $\vecc{x}$, and
one can easily see that
\begin{equation*}
\mu^{\otimes n}(\vecc{x})=e^{-n\bz\sr{\type{x}}{\mu}+S(\type{x}) \jz}.
\end{equation*}
In particular,
\begin{equation}\label{product probability}
\type{x}^{\otimes n}(\vecc{x})=e^{-nS(\type{x})}, \ds\ds\text{and}\ds\ds
\mu^{\otimes n}(\vecc{x})=\type{x}^{\otimes n}(\vecc{x}) e^{-n\bz\sr{\type{x}}{\mu}\jz}.
\end{equation}
A variant of the following bound can be found in \cite{Hoeffding}. For readers' convenience we provide a complete proof here.
\begin{lemma}\label{lemma:the type probability of a type}
Let $\vecc{x}\in\X^n$ and $r:=|\supp\type{x}|$. Then,
\begin{equation*}
\frac{1}{n}\log \type{x}^{\otimes n}\bz\{\vecc{y}\,:\,\type{y}=\type{x}\}\jz\ge
-\frac{r-1}{2}\frac{\log n}{n}+\frac{r}{n}\bz\log(\sqrt{r/2\pi})-1/12\jz
+\frac{1}{n(12n+1)}.
\end{equation*}
\end{lemma}
\begin{proof}
Let $z_1,\ldots,z_{r}$, be an ordering of the elements of $\supp\type{x}$, and let
$k_i:=n\type{x}(z_i)$. Then
\begin{equation*}
|\{\vecc{y}\,:\,\type{y}=\type{x}\}|=\frac{n!}{k_1!\cdot\ldots\cdot k_{r}!},\ds\ds\ds\ds
\type{x}^{\otimes n}(\vecc{y})=\prod_{i=1}^{r}(k_i/n)^{k_i}, \ds\type{y}=\type{x}.
\end{equation*}
By Stirling's formula (see, e.g., \cite{Feller1}),
\begin{equation*}
(m/e)^m\sqrt{2\pi m}\,e^{1/(12m+1)}\le m!\le (m/e)^m\sqrt{2\pi m}\,e^{1/12m},
\end{equation*}
and hence,
\begin{align*}
p_n&:=\type{x}^{\otimes n}\bz\{\vecc{y}\,:\,\type{y}=\type{x}\}\jz=
|\{\vecc{y}\,:\,\type{y}=\type{x}\}|\type{x}^{\otimes n}(\vecc{x})
=
\frac{n!}{n^n}\prod_{i=1}^{r}\frac{k_i^{k_i}}{k_i!}\\
&\ge
 e^{-n}\,\sqrt{2\pi n}\,e^{1/(12n+1)}\prod_{i=1}^{r}e^{k_i}\sqrt{2\pi k_i}\inv e^{-1/12k_i}\\
&=
\frac{\sqrt{2\pi n}}{\sqrt{2\pi}^{r}\sqrt{k_1\cdot\ldots\cdot k_{r}}}
\exp(1/(12n+1)-1/12k_1-\ldots-1/12k_r).
\end{align*}
Using $\sqrt[r]{k_1\cdot\ldots\cdot k_r}\le\frac{k_1+\ldots+k_r}{r}=\frac{n}{r}$, we have
$\sqrt{k_1\cdot\ldots\cdot k_r}\le (n/r)^{r/2}$, while $k_i\ge 1$ yields
$1/k_1+\ldots+1/k_r\le r$, and hence,
\begin{align*}
p_n&\ge
\frac{\sqrt{2\pi n}}{\sqrt{2\pi}^{r}}(r/n)^{r/2}
\exp\bz\frac{1}{12n+1}-\frac{r}{12}\jz
\ge
(\sqrt{r/2\pi})^{r}n^{1/2-r/2}
\exp\bz\frac{1}{12n+1}-\frac{r}{12}\jz,
\end{align*}
which yields
\begin{equation*}
\frac{1}{n}\log p_n\ge
-\frac{r-1}{2}\frac{\log n}{n}+\frac{r}{n}\bz\log(\sqrt{r/2\pi})-1/12\jz
+\frac{1}{n(12n+1)}.\qedhere
\end{equation*}
\end{proof}

Let $\T_n$ denote the collection of all types arising from length $n$ sequences, i.e.,
$\T_n:=\{\type{x}\,:\,\vecc{x}\in\X^n\}$.
It is known that $\cup_{n\in\bN}\T_n$ is dense in $\M_1(\X)$, and
$\inf_{\nu\in\T_n}\norm{\mu-\nu}_1\le\frac{|\X|}{n}$ for any $\mu\in\M_1(\X)$; see, e.g.,
\cite{DZ}.
Moreover, the following has been shown in Lemma A.2 of \cite{Hoeffding}:
\begin{lemma}\label{lemma:Hoeffding approximation}
Let $v\in\bR^{\X}$ and $c\in\bR$, and assume that the half-spaces
$H_1:=\{f\in\bR^{\X}\,:\,\sum_{x\in\X}f(x)v(x)<c\}$ and
$H_2:=\{f\in\bR^{\X}\,:\,\sum_{x\in\X}f(x)v(x)>c\}$
have non-trivial intersections with $\M_1(\X)$. Then for every
$\mu\in\M_1(\X)$ such that $\sum_{x\in\X}\mu(x)v(x)=c$, and every $n\ge r(r-1)$, where $r:=|\supp\mu|$, there exist types $\mu_1\in H_1\cap\T_n$ and $\mu_2\in H_2\cap\T_n$
such that
\begin{equation*}
\max\left\{\norm{\mu-\mu_1}_1,\norm{\mu-\mu_2}_1\right\}\le\frac{2(r-1)}{n}.
\end{equation*}
\end{lemma}

For more about types and their applications in information theory, see e.g., \cite{CsK}.

\section{Optimal Type II errors}\label{sec:typeII}

Consider the state discrimination problem described in the Introduction.
In this section, we will give bounds on the error probabilities
$\beta_{n,\ep}$ and $\beta_{n,e^{-nr}}$.
The key technical tool will be the following lemma about the duality of linear programming, known as \ki{Slater's condition}; for a proof, see Problem 4 in Section 7.2 of
\cite{Barvinok}.
\begin{lemma}\label{lemma:duality}
Let $V_1$ and $V_2$ be real inner product spaces and let $K_i$ be a convex cone in $V_i$.
The dual cone $K_i^*$ is defined as $K_i^*:=\{y\in V_i\,:\,\inner{y}{x}\ge 0,\, x\in K_i\}$.
Let $c\in V_1,\,b\in V_2$ and let $A:\,V_1\to V_2$ be a linear map.
Assume that there exists a $v$ in the interior of $K_1$ such that $Av-b$ is in the interior of $K_2$. Then the following two quantities are equal:
\begin{align*}
\gamma^p:&=\inf\{\inner{c}{v}\,:\,v\ge_{K_1}0,\, Av\ge_{K_2}b\},\\
\gamma^d:&=\sup\{\inner{b}{w}\,:\,w\ge_{K_2^*}0,\, A^*w\le_{K_1^*}c\}.
\end{align*}
\end{lemma}
\medskip

Using Lemma \ref{lemma:duality}, we can give the following alternative characterization of the optimal type II error:
\begin{prop}
For every $\ep\in(0,1)$, we have
\begin{align}
\beta_{1,\ep}&=
\sup_{\lambda\ge 0}\{(1-\ep)\lambda-\Tr(\lambda\rho-\sigma)_+\}=
\sup_{\lambda\ge 0}\left\{\frac{\lambda+1}{2}-\half\norm{\lambda\rho-\sigma}_1-\lambda\ep\right\}\label{variational expression}\\
&\le
\sup_{\lambda\ge 0}\{\lambda^t\Tr\rho^t\sigma^{1-t}-\lambda\ep\},\ds\ds\ds
t\in[0,1].\label{variational bound}
\end{align}
Moreover, for every $n\in\bN$ and every $t\in[0,1)$,
\begin{equation}\label{beta ep upper bound}
\frac{1}{n}\log\beta_{n,\ep}\le-\rsr{\rho}{\sigma}{t}+\frac{\log\ep\inv}{n}\frac{t}{1-t}-\frac{1}{n}\frac{h_2(t)}{1-t},
\end{equation}
where $h_2(t):=-t\log t-(1-t)\log(1-t),\,t\in[0,1]$.
\end{prop}
\begin{proof}
Let $\tilde\rho$ and $\tilde\sigma$ be density operators on some finite-dimensional Hilbert space, and for each $\ep>0$ define
\begin{equation*}
\beta_\ep:=\min\{\Tr\tilde\sigma T\,:\,0\le T\le I,\,\Tr\tilde\rho (I-T)\le\ep\},
\end{equation*}
which is the optimal type II error for discriminating between $\tilde\rho$ and $\tilde\sigma$ under the constraint that the type I error doesn't exceed $\ep$.
We apply Lemma \ref{lemma:duality} to give an alternative expression for $\beta_\ep$.
To this end, we define
\begin{equation*}
V_1:=\B(\hil)_{sa}, \ds\ds c:=\tilde\sigma,\ds\ds\ds V_2:=\B(\hil)_{sa}\oplus \bR,\ds\ds b:=-I\oplus(1-\ep),
\end{equation*}
where $\B(\hil)_{sa}$ is the real linear vector space of self-adjoint operators on $\hil$.
We equip both $V_1$ and $V_2$ with the Hilbert-Schmidt inner product, and define $K_1$ and $K_2$ to be
the self-dual cones of the positive semidefinite operators.
If we define $A$ to be $A:\,X\mapsto -X\oplus\Tr\tilde\rho X$ then $A^*$ is given by
$A^*:\,X\oplus\lambda\mapsto -X+\lambda\tilde\rho$, and
we see that $\gamma^p=\beta_\ep$.
It is easy to verify that the condition of Lemma \ref{lemma:duality} is satisfied in this case, and hence
\begin{equation*}
\beta_\ep=\gamma^p=\gamma^d=\sup\{-\Tr X+\lambda(1-\ep)\,:\,X\ge 0,\,\lambda\ge 0,\,-X+\lambda\tilde\rho\le\tilde\sigma\}.
\end{equation*}
For a fixed $\lambda\ge 0$, we have
\begin{equation*}
\inf\{\Tr X\,:\,X\ge 0,\,\lambda\tilde\rho-\tilde\sigma\le X\}=\Tr(\lambda\tilde\rho-\tilde\sigma)_+=\half\Tr(\lambda\tilde\rho-\tilde\sigma)+
\half\norm{\lambda\tilde\rho-\tilde\sigma}_1=
\frac{\lambda-1}{2}+\half\norm{\lambda\tilde\rho-\tilde\sigma}_1
\end{equation*}
(the first identity can also be seen by a duality argument). Hence, we have
\begin{align*}
\beta_\ep&=\gamma^d=
\sup_{\lambda\ge 0}\{(1-\ep)\lambda-\Tr(\lambda\tilde\rho-\tilde\sigma)_+\}=
\sup_{\lambda\ge 0}\left\{\frac{\lambda+1}{2}-\half\norm{\lambda\tilde\rho-\tilde\sigma}_1-\lambda\ep\right\}\\
&\le
\sup_{\lambda\ge 0}\{\lambda^t\Tr\tilde\rho^t\tilde\sigma^{1-t}-\lambda\ep\},
\ds\ds\ds t\in[0,1],
\end{align*}
where the last inequality is due to Lemma \ref{lemma:Aud}.
Choosing $\tilde\rho=\rho$ and $\tilde\sigma=\sigma$ gives \eqref{variational expression} and
\eqref{variational bound}.

Note that $f(\lambda):=\lambda^t\Tr\tilde\rho^t\tilde\sigma^{1-t}-\lambda\ep$ is concave, and hence if
$f(\lambda)$ has a stationary point $\lambda^*$ then this is automatically a global maximum.
Solving $f'(\lambda^*)=0$ in the case $t\ne 1$, we get
\begin{equation*}
\lambda^*=\bz\frac{t\Tr\tilde\rho^t\tilde\sigma^{1-t}}{\ep}\jz^{\frac{1}{1-t}},
\end{equation*}
and substituting it back, we get
\begin{equation*}
\log\beta_\ep\le\log
f(\lambda^*)=-\frac{t\log\ep-\log\Tr\tilde\rho^t\tilde\sigma^{1-t}}{1-t}-\frac{h_2(t)}{1-t},\ds\ds\ds t\in[0,1).
\end{equation*}
Choosing now $\tilde\rho=\rho^{\otimes n},\,\tilde\sigma=\sigma^{\otimes n}$, we obtain
\begin{equation}\label{beta r upper3}
\frac{1}{n}\log\beta_{n,\ep}\le
-\frac{(t/n)\log\ep-\log\Tr\rho^t\sigma^{1-t}}{1-t}-\frac{1}{n}\frac{h_2(t)}{1-t},\ds\ds\ds t\in[0,1),
\end{equation}
which is equivalent to \eqref{beta ep upper bound}.
\end{proof}

\begin{thm}\label{thm:type II bounds}
For every $\ep\in(0,1)$ and every $n\in\bN$, we have
\begin{align}
\frac{1}{n}\log\beta_{n,\ep}&\le-\sr{\rho}{\sigma}+\frac{1}{\sqrt{n}}4\sqrt{2}\log\ep\inv\log\eta-\frac{2\log 2}{n},\label{beta ep upper bound1}\\
\frac{1}{n}\log\beta_{n,\ep}&\ge-\sr{\rho}{\sigma}-\frac{1}{\sqrt{n}}4\sqrt{2}\log(1-\ep)\inv\log\eta,\label{beta ep lower bound}
\end{align}
where $\eta:=1+e^{\half\rsr{\rho}{\sigma}{3/2}}+e^{-\half\rsr{\rho}{\sigma}{1/2}}$,
as in \eqref{eta def}.
Moreover, for every $n\in\bN$ and every $r>-\log\Tr\rho\sigma^0$ we have
\begin{align}\label{beta r upper}
\frac{1}{n}\log\beta_{n,e^{-nr}}&\le-\hdist{\rho}{\sigma}{r}-\frac{1}{n}\frac{h_2(t_r)}{1-t_r},
\end{align}
where
$t_r:=\argmax_{0\le t<1}\left\{\frac{-tr-\log\Tr\rho^t\sigma^{1-t}}{1-t}\right\}$,
and $t_r>0\iff r<-\psi(0)-\psi'(0)$.
\end{thm}
\begin{proof}
The upper bound \eqref{beta ep upper bound} with the choice
$\ep=e^{-nr}$ yields
\begin{equation}\label{beta r upper2}
\frac{1}{n}\log\beta_{n,r}\le
-\frac{-tr-\log\Tr\rho^t\sigma^{1-t}}{1-t}-\frac{1}{n}\frac{h_2(t)}{1-t},\ds\ds\ds t\in[0,1).
\end{equation}
If $r>-\log\Tr\rho\sigma^0$ then there exists a $t_r\in[0,1)$ such that
\begin{equation*}
\frac{-rt_r-\log\Tr\rho^{t_r}\sigma^{1-t_r}}{1-t_r}=
\max_{0\le t<1}\frac{-rt-\log\Tr\rho^{t}\sigma^{1-t}}{1-t}=
\hdist{\rho}{\sigma}{r}.
\end{equation*}
This follows from Lemma \ref{lemma:Hoeffding representation} when
$r<-\psi(0)-\psi'(0)$, where $\psi(t):=\log\Tr\rho^t\sigma^t$, and for
$r\ge-\psi(0)-\psi'(0)$ we have $t_r=0$.
With this $t_r$, \eqref{beta r upper2} yields \eqref{beta r upper}.

Next, we apply Lemma \ref{lemma:Renyi bounds}
with $A:=\rho$ and $B:=\sigma$
to the upper bound \eqref{beta ep upper
bound} to get
\begin{align*}
\frac{1}{n}\log\beta_{n,\ep}&\le
-\sr{\rho}{\sigma}+(1-t)4(\cosh c)(\log\eta)^2+\frac{-\log\ep}{n}\frac{t}{1-t}-\frac{1}{n}\frac{h_2(t)}{1-t}\\
&\le
-\sr{\rho}{\sigma}+(1-t)4
(\cosh c)(\log\eta)^2+\frac{-\log\ep}{n}\frac{1}{1-t}-\frac{2\log2}{n},
\end{align*}
which is valid for $1-\delta\le t<1$.
Now let us choose $t=1-a/\sqrt{n}$ for some $a>0$; then we have
\begin{align*}
\frac{1}{n}\log\beta_{n,\ep}&\le
-\sr{\rho}{\sigma}+\frac{a}{\sqrt{n}}4
(\cosh c)(\log\eta)^2+\frac{-\log\ep}{\sqrt{n}}\frac{1}{a}-\frac{2\log2}{n},
\end{align*}
and optimizing over $a$ yields
\begin{align}\label{beta ep upper bound3}
\frac{1}{n}\log\beta_{n,\ep}&\le
-\sr{\rho}{\sigma}+\frac{2}{\sqrt{n}}
\sqrt{4(\cosh c)(\log\eta)^2\log\ep\inv}-\frac{2\log2}{n},
\end{align}
where the optimum is reached at $a^*=\sqrt{\frac{-\log\ep}{4(\cosh c)(\log\eta)^2}}$. The
above upper bound is valid as long as
$1-a^*/\sqrt{n}\ge 1-\delta$, or equivalently,
\begin{equation}\label{n lower bounds}
n\ge 4 (a^*)^2=\frac{\log\ep\inv}{(\cosh c)(\log\eta)^2}\ds\ds\text{and}\ds\ds
n\ge 4 (a^*)^2(\log\eta)^2/c^2=\frac{\log\ep\inv}{c^2\cosh c}.
\end{equation}
Let us now choose $c$ such that $\cosh c=2\log\ep\inv$. Then it is easy to see that
$c=\arcosh(2\log\ep\inv)=\log\bz 2\log\ep\inv+\sqrt{(2\log\ep\inv)^2-1}\jz\ge 1$. Since we
also have $\log\eta>1$, we see that both of the lower bounds in
\eqref{n lower bounds} are less than $1$, i.e., the upper bound in
\eqref{beta ep upper bound3} is valid for all $n\in\bN$ with
$\cosh c=2\log\ep\inv$, which yields \eqref{beta ep upper bound1}.

To prove \eqref{beta ep lower bound}, we apply the idea of \cite{Nagaoka2} to use the monotonicity of the R\'enyi relative entropies to get a lower bound on $\beta_{n,\ep}$. Let $T$ be any test such that
$\alpha_n(T)=\Tr\rho_n(I-T)\le\ep$; then for every $t\in(1,2]$ we have
\begin{align*}
\Tr\rho_n^t\sigma_n^{1-t}&\ge
(\Tr\rho_n T)^t(\Tr\sigma_n T)^{1-t}+(\Tr\rho_n (I-T))^t(\Tr\sigma_n (I-T))^{1-t}\\
&\ge
(\Tr\rho_n T)^t(\Tr\sigma_n T)^{1-t}
\ge
(1-\ep)^t(\Tr\sigma_n T)^{1-t}.
\end{align*}
Taking the logarithm and rearranging then yields
\begin{align*}
\log\Tr\sigma_n T\ge -\rsr{\rho_n}{\sigma_n}{t}-\frac{t}{t-1}\log(1-\ep)\inv.
\end{align*}
Taking now the infimum over all $T$ such that $\alpha_n(T)\le\ep$, and using Lemma \ref{lemma:Renyi bounds}, we obtain
\begin{align*}
\frac{1}{n}\log\beta_{n,\ep}&
\ge-\frac{1}{n}\rsr{\rho_n}{\sigma_n}{t}-\frac{1}{n}\frac{t}{t-1}\log(1-\ep)\inv\\
&\ge
-\sr{\rho}{\sigma}-(4\cosh c)(t-1)(\log \eta)^2-\frac{1}{n}\frac{1}{t-1}\log(1-\ep)\inv.
\end{align*}
Again, let $t:=1+a/\sqrt{n}$; then
\begin{align*}
\frac{1}{n}\log\beta_{n,\ep}&
\ge
-\sr{\rho}{\sigma}-\frac{1}{\sqrt{n}}\bz a(4\cosh c)(\log \eta)^2+\frac{\log(1-\ep)\inv}{a}\jz,
\end{align*}
and optimizing over $a$ yields
\begin{align*}
\frac{1}{n}\log\beta_{n,\ep}&
\ge
-\sr{\rho}{\sigma}-\frac{2}{\sqrt{n}}
\sqrt{(4\cosh c)(\log \eta)^2\log(1-\ep)\inv},
\end{align*}
where the optimum is reached at $a^*=\sqrt{\frac{\log(1-\ep)\inv}{4(\cosh c)(\log\eta)^2}}$.
This bound is valid as long as $1<t<1+\delta$, or equivalently, if
\begin{equation*}
n\ge 4 (a^*)^2=\frac{\log(1-\ep)\inv}{(\cosh c)(\log\eta)^2}\ds\ds\text{and}\ds\ds
n\ge 4 (a^*)^2(\log\eta)^2/c^2=\frac{\log(1-\ep)\inv}{c^2\cosh c}.
\end{equation*}
Choosing $c=\arcosh(2\log(1-\ep)\inv)$, the same argument as above leads to \eqref{beta ep lower bound}.
\end{proof}

\begin{rem}
The bounds in \eqref{beta ep upper bound1} and \eqref{beta ep lower bound} yield immediately the quantum Stein's lemma, i.e., Theorem \ref{thm:Stein}.
\end{rem}

\begin{rem}
For any chosen pair of states $\rho$ and $\sigma$,
the set of points $\{(\alpha(T),\beta(T))\,:\,T\text{ test}\}$
forms a convex set, which we call the error set here,
and the lower boundary of this set is what constitutes the sought-after optimal errors.
It is easy to see that
for any $\ep\in(0,1)$, $\beta_\ep:=\min\{\Tr\sigma T\,:\,\Tr\rho(I-T)\le \ep\}$ can be
attained at a test for which $\Tr\rho(I-T)=\ep$. It is also easy to see that
there exists a $\lambda_\ep\ge 0$ and a Neyman-Pearson test $T_\ep$ such that
$\{\lambda_\ep\rho-\sigma>0\}\le T_\ep\le\{\lambda_\ep\rho-\sigma\ge 0\}$, for which
$\Tr\rho(I-T_\ep)=\ep$, and by the Neyman-Pearson lemma (see the Introduction),
$\beta_\ep=\Tr\sigma T_\ep$.
That is, all points on the lower boundary can be attained by Neyman-Pearson tests.
Finally, we have the identity
$\Tr\sigma T_\ep=\lambda_\ep\Tr\rho T_\ep-\Tr(\lambda_\ep\rho-\sigma)_+=
\lambda_\ep(1-\ep)-\Tr(\lambda_\ep\rho-\sigma)_+$; cf.~formula
\eqref{variational expression}. Here, $\lambda$ is related to the slope of the tangent
line of the lower boundary at the point $(\alpha(T_\ep),\beta(T_\ep))$.
In the next section we follow a different approach to scale the lower boundary of the error
set by essentially fixing the slope of the tangent line and looking for the optimal errors
corresponding to that slope; this is reached by minimizing the mixed error probabilities
$e^{-na}\alpha_n(T)+\beta_n(T)$.
\end{rem}

\section{The mixed error probabilities}\label{sec:mixed}

Consider again the state discrimination problem described in the
Introduction. For every $a\in\bR$, let $e_n(a)$ be the mixed error probability as defined
in \eqref{mixed error def}, and let $\vfi(a)$ and $\hat\vfi(a)$ be as in \eqref{phi}.
Note that $e_n(0)$ is twice the Chernoff error with equal priors $p=1-p=1/2$, and
for every $r>-\log\Tr\rho\sigma^0$,
we have
$\vfi(a_r)=\hdist{\rho}{\sigma}{r}$ and $\hat\vfi(a_r)=r$ for
$a_r:=\hdist{\rho}{\sigma}{r}-r$, due to Lemma \ref{lemma:Hoeffding representation}.

Lemma \ref{lemma:Aud} yields various upper bounds on the error probabilities. These have
already been obtained in
\cite{Aud,HMO2,Nagaoka}. We repeat them here for completeness.
\begin{prop}\label{prop:upper bounds}
For every $a\in\bR$ and every $n\in\bN$, we have
\begin{equation}\label{e_n upper bound}
\frac{1}{n}\log e_n(a)\le-\vfi(a),
\end{equation}
which in turn yields
\begin{align}\label{alpha beta upper bounds}
\frac{1}{n}\log\alpha_n(T)&\le-\hat\vfi(a),\ds\ds\ds\ds\ds\ds
\frac{1}{n}\log\beta_n(T)\le-\vfi(a)
\end{align}
for every $T\in\N_{n,a}$.
In particular, we have
\begin{equation*}
\frac{1}{n}\log e_n(0)\le-\chdist{\rho}{\sigma}
\end{equation*}
for the Chernoff error, and if $r>-\log\Tr\rho\sigma^0$ then we have
\begin{align*}
\frac{1}{n}\log e_n(a_r)&\le-\hdist{\rho}{\sigma}{r},\ds\ds\ds\text{and}\\
\frac{1}{n}\log\alpha_n(T)&\le-\hat\vfi(a_r)=-r,\ds\ds\ds\ds\ds\ds
\frac{1}{n}\log\beta_n(T)\le-\vfi(a_r)=-\hdist{\rho}{\sigma}{r}
\end{align*}
for every $T\in\N_{n,a_r}$, where $a_r=\hdist{\rho}{\sigma}{r}-r$.
\end{prop}
\begin{proof}
For fixed $a\in\bR$ and $n\in\bN$ let $T\in\N_{n,a}$. Then we have
\begin{align}\label{e_n upper bound2}
e_n(a)=e^{-na}\alpha_n(T)+\beta_n(T)=\frac{1+e^{-na}}{2}-\half\norm{e^{-na}\rho_n-\sigma_n}_1\le
e^{-nta}\Tr\rho_n^t\sigma_n^{1-t}
\end{align}
for every $t\in[0,1]$,
where the inequality is due to Lemma \ref{lemma:Aud}.
Since
$\Tr\rho_n^t\sigma_n^{1-t}=\bz\Tr\rho^t\sigma^{1-t}\jz^n$,
taking the infimum over $t\in[0,1]$ in \eqref{e_n upper bound2} yields \eqref{e_n upper
bound}. The inequalities in \eqref{alpha beta upper bounds} are immediate from
$e^{-na}\alpha_n(T)\le e_n(a)$ and $\beta_n(T)\le e_n(a)$. The rest of the assertions
follow as special cases.
\end{proof}
\medskip

To obtain lower bounds on the mixed error probabilities, we will use the mapping described
in the beginning of Section \ref{sec:preliminaries} with $A:=\rho$ and $B:=\sigma$.
Hence, we use the notation $\X:=\X_{\rho,\sigma}$, $p:=p_{\rho,\sigma}$
and $q:=q_{\rho,\sigma}$. Note that $\supp p=\supp q=\X$ and
$p(\X)\le 1,\,q(\X)\le 1$. For every $a\in\bR$ and $n\in\bN$, let
\begin{equation*}
\tilde e_n(a):=\min\{e^{-na}p^{\otimes n}(\X^n\setminus T)+q^{\otimes
n}(T)\,:\,T\subset\X^n\}.
\end{equation*}
It is easy to see that
\begin{equation}\label{classical mixed error}
\tilde e_n(a)=e^{-na}p^{\otimes n}(\X^n\setminus N_{n,a})+q^{\otimes n}(N_{n,a}),
\end{equation}
where
\begin{equation*}
N_{n,a}:=\left\{\vecc{x}\in\X^n\,:\,\frac{1}{n}\log\frac{p^{\otimes n}(\vecc{x})}{q^{\otimes n}(\vecc{x})}\ge a\right\}
\end{equation*}
is a classical Neyman-Pearson test for discriminating between $p$ and $q$.
One can easily verify that
$N_{n,a}=\{\vecc{x}\in\X^n\,:\,\type{x}\in \N_{a}\}$, where
\begin{equation*}
\N_a=\{\mu\in\M_1(\X)\,:\,\sr{\mu}{q}-\sr{\mu}{p}\ge a\}=
\left\{\mu\in\M_1(\X)\,:\,\sum_{y\in\X} \mu(y)\log\frac{p(y)}{q(y)}\ge a\right\}
\end{equation*}
is the intersection of $\M_1(\X)$ with the half-space $\{f\in\bR^{\X}\,:\,\sum_y f(y)v(y)\ge a\}$, where $v$ is the normal vector $v(y):=\log\frac{p(y)}{q(y)},\,y\in\X$.
We also define $\partial N_a:=\{\mu\in\M_1(\X)\,:\,\sr{\mu}{q}-\sr{\mu}{p}= a\}$.

The following Lemma has been shown in \cite{NSz} (see also Theorem 3.1 in \cite{HMO2} for a slightly different proof):
\begin{lemma}\label{lemma:NSz}
For every $a\in\bR$ and $n\in\bN$, we have $2e_n(a)\ge\tilde e_n(a)$.
\end{lemma}

Hence, in order to give lower bounds on the mixed error probabilities $e_n(a)$, it is
enough to find lower bounds on $\tilde e_n(a)$.
Let $\X^{\infty}:=\displaystyle{\times_{k=1}^{+\infty}\X}$ be equipped with the sigma-field generated by
the cylinder sets, and let
$Y_k(\vecc{x}):=\log\frac{p(x_k)}{q(x_k)},\,\vecc{x}\in\X^{\infty},\,k\in\bN$.
Then $Y_1,Y_2,\ldots$, is a sequence of i.i.d.~random variables on $\X^{\infty}$ with
respect to any product measure. By \eqref{classical mixed error}, we have
\begin{equation*}
\tilde e_n(a)=e^{-na}\tilde\alpha_n(a)+\tilde\beta_n(a),
\end{equation*}
where $\tilde\alpha_n(a):=p^{\otimes n}(\X^n\setminus N_{n,a})$ and
$ \tilde\beta_n(a):=q^{\otimes n}(N_{n,a})$, or equivalently,
\begin{align*}
\tilde\alpha_n(a)=p^{\otimes n}\bz\frac{1}{n}\sum_{k=1}^n Y_k<a\jz,\ds\ds\ds
\tilde\beta_n(a)=q^{\otimes n}\bz\frac{1}{n}\sum_{k=1}^n Y_k\ge a\jz.
\end{align*}
Note that with $\hat p:=p/p(\X)$ and $\hat q:=q/\hat q(\X)$, we have
\begin{equation*}
\Exp_{\hat p}Y_1=\sr{p}{q}/p(\X)=\psi'(1),\ds\ds\ds
\Exp_{\hat q}Y_1=-\sr{q}{p}/q(\X)=\psi'(0).
\end{equation*}
Hence, by the theory of large deviations, $\tilde\alpha_n(a)$ and $\tilde\beta_n(a)$ decay
exponentially fast in $n$ when $\psi'(0)<a<\psi'(1)$. Using Theorem 1 in \cite{BR}, we
can obtain more detailed information about the speed of decay:
\begin{prop}\label{prop:BR bounds}
For every $\psi'(0)<a<\psi'(1)$, there exist constants
$c_1,c_2,d_1,d_2$, depending on $\rho,\sigma$ and $a$, such that
for every $n\in\bN$,
\begin{align}
-\hat\vfi(a)-\half\frac{\log n}{n}+\frac{c_1}{n}&\le\frac{1}{n}\log\tilde\alpha_{n}(a)\le
-\hat\vfi(a)-\half\frac{\log n}{n}+\frac{c_2}{n},\label{BR alpha}\\
-\vfi(a)-\half\frac{\log n}{n}+\frac{d_1}{n}&\le\frac{1}{n}\log\tilde\beta_n(a)\le
-\vfi(a)-\half\frac{\log n}{n}+\frac{d_2}{n}.\label{BR beta}
\end{align}
\end{prop}
\begin{proof}
Note that the moment generating function of $Y_1$ with respect to $\hat q$ is
$M(t):=\Exp_{\hat q}\bz e^{tY_1}\jz=\sum_{x\in\X}p(x)^tq(x)^{1-t}/q(\X)$, and
hence $\inf_{t\in\bR}e^{-ta}M(t)=e^{-\vfi(a)-\log q(\X)}=:\rho_a$. The bounds in
\eqref{BR beta} then follow immediately from Theorem 1 in \cite{BR}, and the bounds in
\eqref{BR alpha} can be proven exactly the same way.
\end{proof}

\begin{rem}\label{rem:BR bounds}
It is easy to see that $\psi'(0)<a<\psi'(1)$ if and only if there exists an $r$ such that
$-\psi(1)<r<-\psi(0)-\psi'(0)$ and $a=a_r$. Hence, Proposition \ref{prop:BR bounds} can be
reformulated in the following way: For every
$-\psi(1)<r<-\psi(0)-\psi'(0)$, there exist constants $\gamma_1,\gamma_2,\delta_1,\delta_2$,
depending on $\rho,\sigma$ and $r$,
such that for every $n\in\bN$,
\begin{align*}
-r-\half\frac{\log n}{n}+\frac{\gamma_1}{n}&\le\frac{1}{n}\log\tilde\alpha_{n,r}\le
-r-\half\frac{\log n}{n}+\frac{\gamma_2}{n},
\\
-H_r-\half\frac{\log n}{n}+\frac{\delta_1}{n}&\le\frac{1}{n}\log\tilde\beta_{n,r}\le
-H_r-\half\frac{\log n}{n}+\frac{\delta_2}{n},
\end{align*}
where $\tilde\alpha_{n,r}:=\alpha_n(a_r),\,\tilde\beta_{n,r}:=\beta_n(a_r)$.
\end{rem}

\begin{cor}\label{cor:mixed BR lower bounds}
For every $\psi'(0)<a<\psi'(1)$, there exists a constant
$c$, depending on $\rho,\sigma$ and $a$, such that
for every $n\in\bN$,
\begin{align*}
\frac{1}{n}\log e_n(a)\ge
-\vfi(a)-\half\frac{\log n}{n}+\frac{c}{n}.
\end{align*}
In particular, if $\psi'(0)<0<\psi'(1)$ then
\begin{align*}
\frac{1}{n}\log e_n(0)\ge
-\chdist{\rho}{\sigma}-\half\frac{\log n}{n}+\frac{c}{n}.
\end{align*}

Equivalently, for every
$-\psi(1)<r<-\psi(0)-\psi'(0)$, there exists a constant $\gamma$,
depending on $\rho,\sigma$ and $r$,
such that for every $n\in\bN$,
\begin{align*}
\frac{1}{n}\log e_n(a_r)\ge
-H_r-\half\frac{\log n}{n}+\frac{\gamma}{n}.
\end{align*}
\end{cor}
\begin{proof}
Immediate from Lemma \ref{lemma:NSz}, Proposition \ref{prop:BR bounds} and Remark
\ref{rem:BR bounds}
\end{proof}
\medskip

Proposition \ref{prop:BR bounds} and Remark \ref{rem:BR bounds}
show the following: In the classical case, the leading term
in the deviation of the logarithm of the type I and type II errors from their asymptotic
values are  exactly $-\half\frac{\log n}{n}$. Using Lemma \ref{lemma:NSz}, we can obtain
lower bounds on the mixed error probabilities in the quantum case with the same leading
term, as shown in Corollary \ref{cor:mixed BR lower bounds}. Unfortunately, this method does
not make it possible to obtain upper bounds
on the mixed quantum errors, or bounds on the individual quantum errors. Another drawback of
the above bounds is that the constants in the $1/n$ term depend on $a$ (or $r$) in a very
complicated way, and hence it is difficult to see whether for small $n$ it is actually
the $\frac{\log n}{n}$ term or the $1/n$ term that dominates the deviation. Below we give
similar lower bounds on the classical type I and type II errors, and hence also on the mixed
quantum errors, where all constants are parameter-independent and easy to evaluate, on the
expense of increasing the constant before the $\frac{\log n}{n}$ term.
To reduce redundancy, we formulate the bounds only for
$\tilde\alpha_{n,r}$ and $\tilde\beta_{n,r}$; the corresponding bounds for
$\tilde\alpha_n(a)$ and $\tilde\beta_n(a)$ follow by an obvious reformulation.

\begin{prop}\label{prop:classical lower bound}
For every $-\psi(1)<r<-\psi(0)-\psi'(0)$ and $n\ge |\X|(|\X|-1)$,
\begin{align}
\frac{1}{n}\log\tilde\alpha_{n,r}&\ge
-r-\frac{3(|\X|-1)}{2}\frac{\log n}{n}-\frac{c_n}{n}+\frac{1}{n(12n+1)},\label{alpha}\\
\frac{1}{n}\log\tilde\beta_{n,r}&\ge
-H_r-\frac{3(|\X|-1)}{2}\frac{\log n}{n}-\frac{d_n}{n}+\frac{1}{n(12n+1)},\label{beta}
\end{align}
where
$c_n$ in \eqref{alpha}
can be upper bounded as
\begin{equation*}
c_n\le (|\X|-1)(1+\log p_{\min}^{-2})+1.3,
\end{equation*}
and for large enough $n$,
\begin{equation*}
c_n= (|\X|-1)(1+\log p_{\min}^{-2})
-|\X|\bz\log\sqrt{|\X|/2\pi}-1/12\jz,
\end{equation*}
where $p_{\min}:=\min_{x\in\X}\{p(x)\}$.
The same statements hold for $d_n$ in \eqref{beta}, with $p_{\min}$ replaced with
$q_{\min}:=\min_{x\in\X}\{q(x)\}$.
\end{prop}
\begin{proof}
The proofs of \eqref{alpha} and \eqref{beta} go exactly the same way; below we prove
\eqref{beta}. Let $t_r$ be as in Lemma \ref{lemma:Hoeffding representation}.
By Lemma \ref{lemma:Hoeffding representation}, we have
$\sr{\mu^{t_r}}{q}-\sr{\mu^{t_r}}{p}=H_r-r=a_r$, and hence $\mu^{t_r}\in\partial\N_{a_r}$.
For a fixed $r$ and $n\ge r(r-1)$, let $\vecc{x}\in\X^n$ be a sequence such that
\begin{equation*}
a_r<\frac{1}{n}\log\frac{p^{\otimes n}(\vecc{x})}{q^{\otimes n}(\vecc{x})}=\sum_{y\in\X}\type{x}(y)\log\frac{p(y)}{q(y)}
\ds\ds\ds\text{and}\ds\ds\ds
\norm{\mu^{t_r}-\type{x}}_1\le\frac{2(|\X|-1)}{n}.
\end{equation*}
The existence of such a sequence is guaranteed by Lemma \ref{lemma:Hoeffding approximation}. Obviously, $\type{x}\in \N_{a_r}$. By \eqref{product probability},
\begin{align*}
\tilde\beta_{n,r}&=q^{\otimes n}\bz\{\vecc{y}\,:\,\type{y}\in\N_{a_r}\}\jz\ge
q^{\otimes n}\bz\{\vecc{y}\,:\,\type{y}=\type{x}\}\jz\\
&=
\type{x}^{\otimes n}\bz\{\vecc{y}\,:\,\type{y}=\type{x}\}\jz e^{-n\sr{\type{x}}{q}}.
\end{align*}
Using then Lemma \ref{lemma:the type probability of a type},
\begin{align}\label{lower bound1}
\frac{1}{n}\log\tilde\beta_{n,r}\ge
-\sr{\type{x}}{q}
-\frac{s_{\vecc{x}}-1}{2}\frac{\log n}{n}+\frac{s_{\vecc{x}}}{n}\bz\log\sqrt{\frac{s_{\vecc{x}}}{2\pi}}-\frac{1}{12}\jz
+\frac{1}{n(12n+1)},
\end{align}
where $s_{\vecc{x}}:=|\supp T_{\vecc{x}}|$.
By Lemma \ref{lemma:Hoeffding representation}, $H_r=\sr{\mu^{t_r}}{q}$, and using Lemma \ref{lemma:Fi} yields,
with $k:=|\X|-1$,
\begin{align*}
|\sr{T_{\vecc{x}}}{q}-H_r|&=
|\sr{T_{\vecc{x}}}{q}-\sr{\mu^{t_r}}{q}|\\
&=
|-S(T_{\vecc{x}})+S(\mu^{t_r})+\sum_y(\mu^{t_r}(y)-T_{\vecc{x}}(y))\log q(y)|\\
&\le
(k/n)\log k+h_2(k/n)-(2k/n)\log q_{\min}.
\end{align*}
Note that $\eta(x):=-x\ln x$ is concave, and hence $\eta(x)\le \eta(1)+\eta'(1)(x-1)=1-x$, which in turn yields
\begin{equation*}
h_2(k/n)=-\frac{k}{n}\log \frac{k}{n}-\bz 1-\frac{k}{n}\jz\log\bz 1-\frac{k}{n}\jz
\le
\frac{k}{n}\log n-\frac{k}{n}\log k+\frac{k}{n},
\end{equation*}
and hence,
\begin{align*}
-\sr{T_{\vecc{x}}}{q}&\ge -H_r
-(k/n)\log k-h_2(k/n)+(2k/n)\log q_{\min}\\
&\ge
-H_r -(k/n)\log k-\frac{k}{n}\log n+\frac{k}{n}\log k-\frac{k}{n}+(2k/n)\log q_{\min}\\
&=
-H_r -\frac{k}{n}\log n-\frac{k}{n}+(2k/n)\log q_{\min}.
\end{align*}
Finally, combining the above lower bound with \eqref{lower bound1}, we obtain
\begin{equation*}
\frac{1}{n}\log\beta_{n,r}\ge
-H_r-\frac{3(|\X|-1)}{2}\frac{\log n}{n}-\frac{c}{n}+\frac{1}{n(12n+1)},
\end{equation*}
where
\begin{equation*}
c=(|\X|-1)(1+\log q_{\min}^{-2})
-|s_{\vecc{x}}|\bz\log\sqrt{s_{\vecc{x}}/2\pi}-1/12\jz.
\end{equation*}
It is easy to see that the lowest value of $f(n):=n\bz \log(\sqrt{n/(2\pi)})-1/12\jz,\,n\in\bN$, is at $n=2$, and is lower bounded by $-1.3$.
Moreover, for large enough $n$, $\supp \type{x}=\supp\mu^{t_r}=\X$, which yields the statements about $c_n$.
\end{proof}

Combining Proposition \ref{prop:classical lower bound} with Lemma \ref{lemma:NSz}, we obtain the following lower bounds on the quantum mixed error  probabilities:

\begin{thm}\label{thm:quantum lower bounds}
Let $d$ be the dimension of the subspace on which $\rho$ and $\sigma$ are supported. For
every $-\psi(1)<r<-\psi(0)-\psi'(0)$ and $n\ge d^2(d^2-1)$, we have
\begin{equation}\label{mixed error lower bound}
\frac{1}{n}\log e_n(a_r)\ge -\hdist{\rho}{\sigma}{r}-\frac{3(d^2-1)}{2}\frac{\log n}{n}-\frac{c}{n}+\frac{1}{n(12n+1)},
\end{equation}
where $c$ is a constant depending only on $\rho$ and $\sigma$.

If, moreover, there exists a $t\in(0,1)$ such that $\psi'(t)=0$ then
\begin{equation*}
\frac{1}{n}\log e_n(0)\ge -\chdist{\rho}{\sigma}-\frac{3(d^2-1)}{2}\frac{\log n}{n}-\frac{c}{n}+\frac{1}{n(12n+1)}.
\end{equation*}
\end{thm}
\begin{proof}
The inequality in \eqref{mixed error lower bound} is immediate from
Lemma \ref{lemma:NSz} and Proposition \ref{prop:classical lower bound},
by taking into account that
$|\supp p\cup\supp q|\}\le d^2$. This bound applies to the Chernoff error, i.e., the
case $a=0$, if $0=a_r=\psi'(t_r)$ for some $-\psi(1)<r<-\psi(0)-\psi'(0)$, which is
equivalent to the existence of a $t\in(0,1)$ such that $\psi'(t)=0$.
\end{proof}

\begin{rem}\label{rem:Hoeffding lower bound}
By the bound given in Proposition \ref{prop:classical lower bound}, the constant $c$ in
Theorem \ref{thm:quantum lower bounds} can be upper bounded as
\begin{equation*}
c\le (d^2-1)(1-2\log\min\{p_{\min},q_{\min}\})+1.3,
\end{equation*}
where
\begin{equation*}
p_{\min}:=\min_{i,j}\{\lambda_i\Tr P_iQ_j\,:\,\Tr P_iQ_j>0\},
\ds\ds\ds
q_{\min}:=\min_{i,j}\{\eta_j\Tr P_iQ_j\,:\,\Tr P_iQ_j>0\},
\end{equation*}
and $\rho=\sum_i\lambda_i P_i,\,\sigma=\sum_j\eta_jQ_j$ are the spectral decompositions of $\rho$ and $\sigma$, respectively.
\end{rem}

\section{Closing remarks}
In this paper we studied the finite-size behaviour of various error probabilities related to binary state discrimination. In the classical case, the error probabilities $\alpha_n(a)$ and $\beta_n(a)$, corresponding to the Neyman-Pearson tests, can be written as large deviation probabilities, and their exponential decay rate is given by Cram\'er's large deviation theorem
\cite{DZ}.
If $p_n(a)$ denotes $\alpha_n(a),\,\beta_n(a)$, or the mixed error probability $e_n(a)$, for some $a\in\bR$, then the upper bound of Cram\'er's large deviation theorem tells that $p_n(a)\le e^{-nI(a)}$, where $I(a)>0$ for the relevant values of $a$. The more refined large deviation theorem of Bahadur and Rao \cite{BR} yields a faster decay, of the form
\begin{equation}\label{largedev bound}
p_n(a)\le \frac{C(a)}{\sqrt{n}}e^{-nI(a)},
\end{equation}
where $C(a)$ is a constant
(depending on $a$ but not on $n$). Moreover, it shows that this bound is optimal in the sense that there exists another constant $c(a)$ such that
$\frac{c(a)}{\sqrt{n}}e^{-nI(a)}\le p_n(a)$.
(See also \cite{Salikhov} for an upper bound on the constant $C(a)$, and \cite{BS} for an extension to correlated random variables.)
By mapping the quantum problem into a classical one, using the method of Nussbaum and Szko\l a \cite{NSz}, one can easily obtain a lower bound on the mixed error probability $e_n(a)$ of the form $\frac{c(a)}{\sqrt{n}}e^{-nI(a)}\le e_n(a)$, as given in Corollary \ref{cor:mixed BR lower bounds}. Unfortunately, with this method it is only possible to obtain a lower bound, and only on the mixed error probabilities $e_n(a)$, and not on the individual error probabilities $\alpha_n(a)$ and $\beta_n(a)$. It shows nevertheless that it is not possible to obtain a faster decay of the mixed error probabilities in the quantum than in the classical case. On the other hand, it remains an open problem whether the optimal decay rate can be attained by using only separable measurements. A different approach to refining Cram\'er's theorem was developed by Hoeffding \cite{Hoeffding}, using the method of types. Although this method yields a somewhat looser lower bound, its advantage is that the constants can be easily bounded by simple expressions that are independent of $a$; see Theorem \ref{thm:quantum lower bounds} and Remark \ref{rem:Hoeffding lower bound} for the quantum versions.

Unlike for the above error probabilities, it is not clear whether the optimal error probabilities $\beta_{n,\ep}$ of Stein's lemma and $\beta_{n,e^{-nr}}$ of the Hoeffding bound can be written as large deviation probabilities for some sequence of random variables. In section \ref{sec:typeII}, we used a linear programming approach to obtain bounds on these error probabilities. Theorem \ref{thm:type II bounds} shows that $\beta_{n,e^{-nr}}\le C(r)e^{-n\hdist{\rho}{\sigma}{r}}$ for some constant $C(r)<1$ which can also be easily evaluated. This bound is clearly not optimal in the classical case, as
$\beta_{n,e^{-nr}}\le\beta_n(a_r)$, and the latter can be upper bounded in the form
$\beta_n(a_r)\le\frac{C(a_r)}{\sqrt{n}}e^{-n\hdist{\rho}{\sigma}{r}}$
(cf.~Proposition \ref{prop:upper bounds} and Remark \ref{rem:BR bounds}).
However, at the moment the bound of Theorem \ref{thm:type II bounds} seems to be the best available one for the quantum case.

To the best of our knowledge, the most detailed information about the asymptotics of
$\beta_{n,\ep}$ so far (even in the classical case) was that
$\lim_{n\to\infty}\frac{1}{n}\log\beta_{n,\ep}=-\sr{\rho}{\sigma}$.
Our bounds in Theorem
\ref{thm:type II bounds} give more detailed information, namely that the deviation of the
error rate $\frac{1}{n}\log\beta_{n,\ep}$ from its limit $-\sr{\rho}{\sigma}$ is at most the
order of $1/\sqrt{n}$, i.e.,
\begin{equation}\label{Stein bounds}
-\frac{f(\ep)}{\sqrt{n}}\le\frac{1}{n}\log\beta_{n,\ep}+\sr{\rho}{\sigma}\le
\frac{g(\ep)}{\sqrt{n}},\ds\ds\ds n\in\bN,
\end{equation}
where
\begin{equation*}
f(\ep)=4\sqrt{2}\log\eta\log(1-\ep)\inv,\ds\ds
g(\ep)=4\sqrt{2}\log\eta\log\ep\inv.
\end{equation*}
Note that here $f(\ep)>0$ and $g(\ep)>0$ for every $\ep\in(0,1)$.
Two questions arise naturally related to the bounds in \eqref{Stein bounds}. The first is whether $1/\sqrt{n}$ is the true order of the deviation. Indeed, it could be possible that the convergence of $\frac{1}{n}\log\beta_{n,\ep}$ to $-\sr{\rho}{\sigma}$ is actually much faster, but still compatible with the bounds in \eqref{Stein bounds}. The second is whether
the upper bound could be improved by replacing $g(\ep)$, which is strictly positive for every $\ep\in(0,1)$, with some negative function $h(\ep)$. Indeed, note that the upper bound in
\eqref{Stein bounds} can be written in the form
\begin{equation*}
\beta_{n,\ep}\le e^{-n\sr{\rho}{\sigma}}e^{g(\ep)\sqrt{n}},
\end{equation*}
i.e., the correction to the exponentially decaying term goes to $+\infty$ as $n\to+\infty$, whereas in \eqref{largedev bound} we obtained a monotonically decaying correction that vanishes asymptotically. The answers to both of these questions can be extracted from the recent paper
\cite{Carl}, as we show below.

Theorem 3 in \cite{Carl} says that for
given (non-identical) states $\rho$ and $\sigma$ with $\supp\rho\le\supp\sigma$,
every $E_2\in\bR$, and every
sequence of measurements $\{T_n,I_n-T_n\}_{n\in\bN}$, if
\begin{equation}\label{sqrt order}
\limsup_{n\to+\infty}\sqrt{n}\bz\frac{1}{n}\log\beta_n(T_n)+\sr{\rho}{\sigma}\jz\le -E_2
\end{equation}
then
\begin{equation*}
\liminf_{n\to+\infty}\alpha_n(T_n)\ge \Phi\bz\frac{E_2}{\sqrt{V(\rho\|\sigma)}}\jz,
\end{equation*}
where $V(\rho\|\sigma):=\Tr\rho\bz\log\rho-\log\sigma\jz^2-\sr{\rho}{\sigma}^2$, and
$\Phi(x)=\frac{1}{\sqrt{2\pi}}\int_0^x e^{-t^2/2}\,dt$ is the cumulative distribution function
of the standard normal distribution. Moreover, there exists a sequence of measurements
$\{T_n,I_n-T_n\}_{n\in\bN}$ such that \eqref{sqrt order} holds, and
\begin{equation*}
\lim_{n\to+\infty}\alpha_n(T_n)=\Phi\bz\frac{E_2}{\sqrt{V(\rho\|\sigma)}}\jz.
\end{equation*}

Consider now all sequences of measurements $\{T_n,I_n-T_n\}_{n\in\bN}$ such that
$\ep(\{T_n\}):=\lim_{n\to+\infty}\alpha_n(T_n)$ exists, and for all such measurements, let
\begin{equation*}
E_2(\{T_n\}):=-\limsup_{n\to+\infty}\sqrt{n}\bz\frac{1}{n}\log\beta_n(T_n)+\sr{\rho}{\sigma}
\jz.
\end{equation*}
The above mentioned results of \cite{Carl}
yield that
\begin{equation}\label{Carl bound}
E_2(\{T_n\})\le \sqrt{V(\rho\|\sigma)}\Phi\inv(\ep(\{T_n\})),
\end{equation}
where the upper bound is sharp. Let $\ep\in(0,1)$ and
for every $n\in\bN$, let $T_{n,\ep}$ be a measurement such that
$\beta_{n,\ep}=\beta_{n}(T_{n,\ep})$. It is easy to see that we can choose $T_{n,\ep}$ such that it also satisfies $\alpha_n(T_{n,\ep})=\ep$; in particular,
$\lim_{n\to+\infty}\alpha_n(T_{n,\ep})=\ep$. It is also easy to see, from the definition of
$\beta_{n,\ep}$ and some simple continuity argument, that
\begin{equation*}
-\limsup_{n\to+\infty}\sqrt{n}\bz\frac{1}{n}\log\beta_{n,\ep}+\sr{\rho}{\sigma}
\jz
\ge
-\limsup_{n\to+\infty}\sqrt{n}\bz\frac{1}{n}\log\beta_n(T_n)+\sr{\rho}{\sigma}
\jz
\end{equation*}
for any sequence of measurements $\{T_n,I_n-T_n\}$ such that
$\ep(\{T_n\})=\ep$. Taking into account the sharpness of the bound in \eqref{Carl bound}, we obtain that
\begin{equation}\label{Carl lower bound}
\limsup_{n\to+\infty}\sqrt{n}\bz\frac{1}{n}\log\beta_{n,\ep}+\sr{\rho}{\sigma}\jz
=
-\sqrt{V(\rho\|\sigma)}\Phi\inv(\ep).
\end{equation}
This shows that the correct order of the deviation of $\frac{1}{n}\log\beta_{n,\ep}$ from
$-\sr{\rho}{\sigma}$ is indeed $1/\sqrt{n}$ (at least for $\ep\ne 1/2$, since then
$\Phi\inv(\ep)\ne 0$). From this we can also conclude that $\beta_{n,\ep}$ cannot be written as a large deviation probability for the ergodic average of a sequence of i.i.d.~random variables, since then the order of the deviation would be $-\half\frac{\log n}{n}$, according to the Bahadur-Rao bound \cite{BR}.

Moreover, \eqref{Carl lower bound} yields that for any $\ep'\in(\ep,1)$ there exist infinitely many
$n\in\bN$ such that
\begin{equation*}
\sqrt{n}\bz\frac{1}{n}\log\beta_{n,\ep}+\sr{\rho}{\sigma}
\jz\ge
-\sqrt{V(\rho\|\sigma)}\Phi\inv(\ep'),
\end{equation*}
or equivalently,
\begin{equation}\label{Carl lower bound2}
\frac{1}{n}\log\beta_{n,\ep}\ge -\sr{\rho}{\sigma}
+\frac{
-\sqrt{V(\rho\|\sigma)}\Phi\inv(\ep')}{\sqrt{n}}.
\end{equation}
In particular, if $\ep<\ep'<1/2$ then
$-\sqrt{V(\rho\|\sigma)}\Phi\inv(\ep')>0$,
and \eqref{Carl lower bound2} shows that
it is not possible to have an upper bound as in \eqref{Stein bounds} with
some $h(\ep)<0$ in place of $g(\ep)$ for $\ep\in(0,1/2)$.

\section*{Acknowledgments}
Partial funding was provided by
the Marie Curie International Incoming Fellowship ``QUANTSTAT'' (MM).
Part of this work was done when MM was a Junior Research Fellow at the Erwin Schr\"odinger
Institute for Mathematical Physics in Vienna and later a Research Fellow in the Centre for Quantum Technologies in Singapore. The Centre for Quantum Technologies is funded
by the Singapore Ministry of Education and the National Research Foundation as part of the
Research Centres of Excellence program.
The authors are grateful for the hospitality of the Institut Mittag-Leffler, Stockholm.
The authors are grateful to an anonymous referee for helpful comments.

\appendix

\renewcommand{\thesection}{Appendix:}

\section{Binary Classical Case}

\renewcommand{\thesection}{\Alph{section}}

In this Appendix we treat the problem of finding sharp upper and lower bounds on the error probability
of discriminating between two binary random variables (r.v.).
One has a distribution $(p,1-p)$, and the other $(q,1-q)$, with $0\le p,q\le 1$.
We assume that both r.v.'s have the same prior probability, namely $1/2$.
We consider the mixed error probability $e_n(a)$ for a Neyman-Pearson test (governed by the parameter $a$) applied to
$n$ identically distributed independent copies of the r.v.'s. This error probability is given by
\be
e_n(a) = \frac{1}{2} \sum_{k=0}^n {n \choose k} \min\left(e^{-na}p^k(1-p)^{n-k}, q^k(1-q)^{n-k}\right).
\label{eq:def_en}
\ee
In the limit of large $n$, this error probability goes to zero exponentially fast, and the rate
$-(\log e_n(a))/n$ tends to $\vfi(a)$ defined as
\be
\vfi(a) = \sup_{0\le t\le 1} \{at-\psi(t)\},\qquad \psi(t)=\log (p^t q^{1-t}+(1-p)^t (1-q)^{1-t}).
\ee
From this function we can derive
the Hoeffding distance between the two distributions:
\be
H_r = \sup_{0\le t< 1} \frac{-rt-\psi(t)}{1-t}.
\ee

Here we are interested in the finite $n$ behaviour of $e_n$, namely at what rate does $-(\log e_n)/n$ itself tend to
its limit. Because we are dealing with binary r.v.'s, $e_n$ is governed by two binomial distributions.
Let $P_{k,n} = {n\choose k}p^k(1-p)^{n-k}$ and $Q_{k,n} = {n\choose k}q^k(1-q)^{n-k}$. By writing the binomial
coefficient in terms of gamma functions, rather than factorials, the values of these distributions can be calculated
for non-integer $k$ (even though these values have no immediate statistical meaning).
We can then solve the equation $e^{-na}P_{k,n} = Q_{k,n}$ for $k$ and get the point
where one term in (\ref{eq:def_en}) becomes
bigger than the second. Let $k=sn$ be that point. Assuming that $p\le q$ we can then rewrite (\ref{eq:def_en}) as
\be
e_n(a) = \frac{1}{2}\left(
\sum_{k=0}^{\lfloor sn\rfloor} {n \choose k} q^k(1-q)^{n-k}
+ \sum_{k=1+\lfloor sn\rfloor}^n {n \choose k}e^{-na}p^k(1-p)^{n-k}
\right).
\label{eq:en2}
\ee
The value of $s$ is the solution of the equation
$$
e^{-na}p^{sn}(1-p)^{(1-s)n} = q^{sn}(1-q)^{(1-s)n},
$$
which is equivalent to
$$
s\log p+(1-s)\log(1-p) - a = s\log q+(1-s)\log(1-q)
$$
hence $s$ is given by
\be
s=s(a) = \frac{\log\frac{1-p}{1-q}-a}{\log\frac{q(1-p)}{p(1-q)}}.
\label{eq:def_sa}
\ee
Alternatively, $s(a)$ is the value of $s$ that minimises (\ref{eq:en2}).

The summations in (\ref{eq:en2}) can be replaced by an integral, each giving rise to a
regularised incomplete beta function,
using the formula for the cumulative distribution function (CDF) of the binomial distribution
\be
\sum_{k=0}^{k_0} {n\choose k}p^k (1-p)^{n-k} = I_{1-p}(n-k_0,k_0+1).\label{eq:CDF}
\ee

The regularised incomplete beta function $I_{z}(k,l)$ is defined as
$$
I_z(k,l) = \frac{B(z,k,l)}{B(k,l)} = \frac{\int_0^z dt\,t^{k-1}(1-t)^{l-1}}{\int_0^1 dt\,t^{k-1}(1-t)^{l-1}}.
$$
We thus get
\be
e_n(a) = \frac{1}{2}\left(I_{1-q}(n-\lfloor sn\rfloor,\lfloor sn\rfloor+1)
+ e^{-na}(1-I_{1-p}(n-\lfloor sn\rfloor,\lfloor sn\rfloor+1))\right).
\label{eq:en3}
\ee
Because $e_n(a)$ is just a summation with summation bounds depending on $n$, as witnessed by the floor
function appearing here, $e_n(a)$ is a non-smooth function of $n$.
To wit, as a function of $n$, $e_n(a)$ exhibits a wave-like pattern, and so does its
rate $-\log(e_n(a))/n$, as shown in Fig.\ \ref{fig2}.
The amplitude and period of these waves increases when $p$ becomes extremely small.
In order to obtain nice bounds on $e_n(a)$,
we will first try and remove the wave patterns by removing the floor function from $e_n(a)$ in a suitable way.
More precisely, we look for upper and lower bounds on $e_n(a)$ that
are as close to $e_n(a)$ as possible.

The complete and incomplete beta functions have certain monotonicity properties.
Since for $t$ between 0 and 1, $t^{k-1}$ decreases with $k$, $B(z,k,l)$ decreases with $k$ and with $l$.
Thus, we immediately get the bounds
\be
B(z,n-sn+1,sn+1) \le B(z,n-\lfloor sn\rfloor,\lfloor sn\rfloor+1) \le B(z,n-sn,sn).
\ee
For the regularised incomplete beta function this means
$$
\frac{B(z,n-sn+1,sn+1)}{B(n-sn,sn)} \le I_z(n-\lfloor sn\rfloor,\lfloor sn\rfloor+1)
\le \frac{B(z,n-sn,sn)}{B(n-sn+1,sn+1)}.
$$
Using the relation $B(k+1,l+1)=\frac{kl}{(k+l)(k+l+1)}B(k,l)$, this yields
$$
\frac{ns(1-s)}{n+1}\,\,I_z(n-sn+1,sn+1) \le I_z(n-\lfloor sn\rfloor,\lfloor sn\rfloor+1)
\le \frac{(n+1)}{ns(1-s)}\,\,I_z(n-sn,sn).
$$

Sharper bounds are obtained by using a monotonicity relation applicable
for the specific arguments appearing here.
Because of relation (\ref{eq:CDF}), we see that $I_{z}(n-x,x)$ is monotonously increasing in $x$
when $x$ is restricted to be an integer between 1 and $n$. It is therefore a reasonable conjecture that
it increases monotonously over all real $x$ such that $0\le x\le n$.
\begin{lemma}
Let $0\le z\le 1$.
The function $x\mapsto I_{z}(n-x,x)$ is monotonously increasing in $x$ for $0\le x\le n$.
\end{lemma}
\textit{Proof.}
The derivative w.r.t.\ $x$ is non-negative provided
$$
B(n-x,x) \frac{d}{dx} B(z,n-x,x) - B(z,n-x,x) \frac{d}{dx} B(n-x,x) \ge 0.
$$
holds.
The derivative of $B(z,n-x,x)$ is given by
$$
\frac{d}{dx} B(z,n-x,x) = \int_0^z dt\,\log((1-t)/t) t^{n-x-1}(1-t)^{x-1}.
$$
Therefore, the derivative of $I_z(n-x,x)$ is non-negative if
\beas
&&\int_0^1 du\, u^{n-x-1}(1-u)^{x-1} \int_0^z dt\,\log((1-t)/t) t^{n-x-1}(1-t)^{x-1} - \\
&&\int_0^z du\, u^{n-x-1}(1-u)^{x-1} \int_0^1 dt\,\log((1-t)/t) t^{n-x-1}(1-t)^{x-1}
\ge0.
\eeas
As both terms have the integral over the area $0\le t,u\le z$ in common, the integrals simplify to
\beas
&&\int_z^1 du\, u^{n-x-1}(1-u)^{x-1} \int_0^z dt\,\log((1-t)/t) t^{n-x-1}(1-t)^{x-1}-\\
&&\int_0^z du\, u^{n-x-1}(1-u)^{x-1} \int_z^1 dt\,\log((1-t)/t) t^{n-x-1}(1-t)^{x-1}.
\eeas
Upon swapping the variables $u$ and $t$ in the second term, this can be rewritten as
\beas
&&\int_z^1 du\, u^{n-x-1}(1-u)^{x-1} \int_0^z dt\,\log((1-t)/t) t^{n-x-1}(1-t)^{x-1}-\\
&&\int_0^z dt\, t^{n-x-1}(1-t)^{x-1} \int_z^1 du\,\log((1-u)/u) u^{n-x-1}(1-u)^{x-1},
\eeas
which simplifies to
$$
\int_z^1 du\, \int_0^z dt\,(\log((1-t)/t)-\log((1-u)/u))\, u^{n-x-1}(1-u)^{x-1}  t^{n-x-1}(1-t)^{x-1}.
$$
Since the integral is over a region where $u\ge t$, and $\log((1-t)/t)-\log((1-u)/u)\ge 0$ for $u\ge t$,
the integral is indeed non-negative.
\qed

Using the lemma, we then get
$$
I_z(n-sn+1,sn) \le I_z(n-\lfloor sn\rfloor,\lfloor sn\rfloor+1) \le I_z(n-sn,sn+1).
$$
This yields upper and lower bounds on $e_n(a)$ given by
\bea
e_n(a) &\ge& (I_{1-q}(n(1-s)+1,ns)+e^{-na}I_{p}(ns+1,n(1-s)))/2 \label{eq:lobound}\\
e_n(a) &\le& (I_{1-q}(n(1-s),ns+1)+e^{-na}I_{p}(ns,n(1-s)+1))/2. \label{eq:upbound}
\eea
Here we have used the relation $1-I_z(a,b) = I_{1-z}(b,a)$.
Numerical computation shows that the large-$n$ behaviour of these bounds are consistent with
the predictions of Proposition \ref{prop:BR bounds}. Two concrete examples are depicted below.
\begin{figure}[ht]
\begin{center}
\includegraphics[width=7cm]{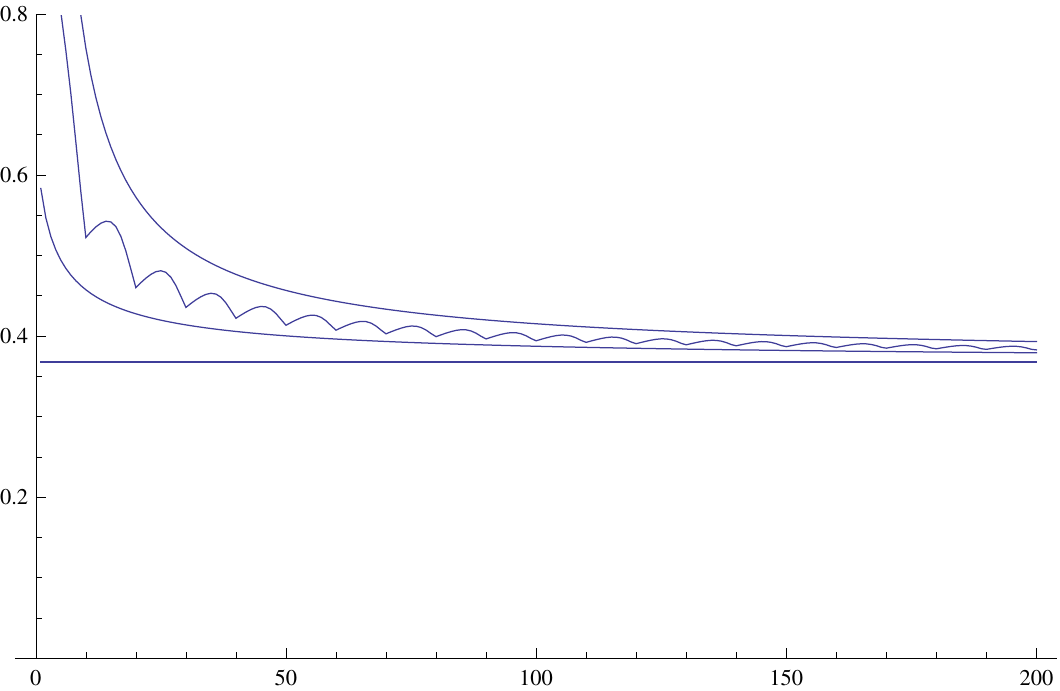}\qquad
\includegraphics[width=7cm]{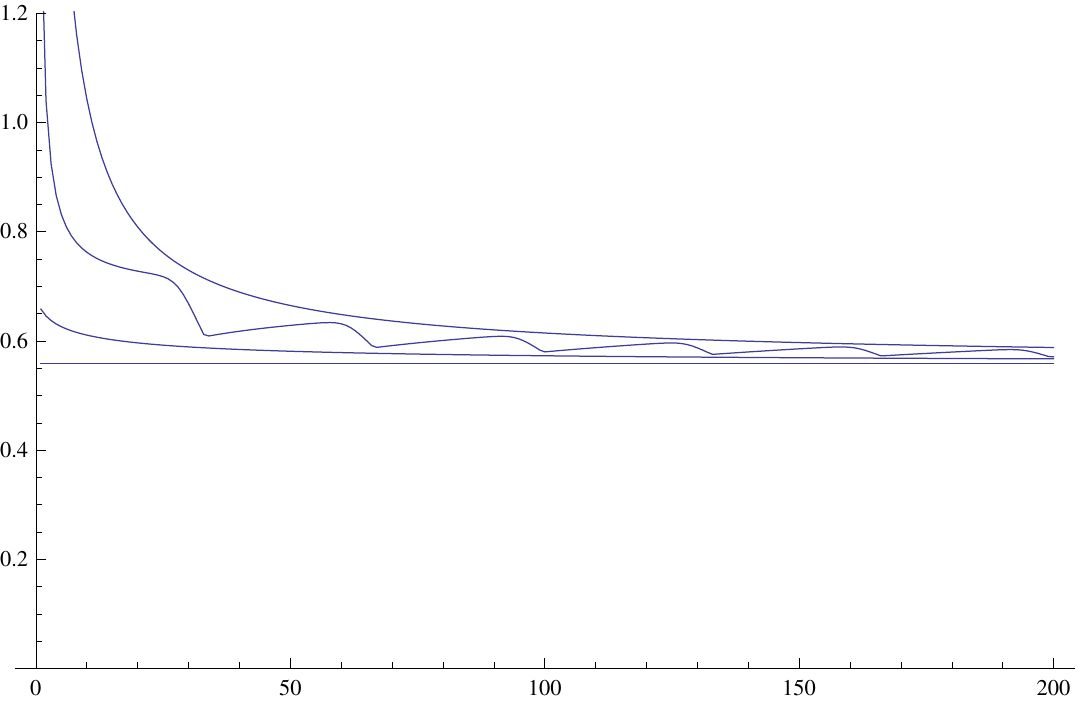}
\caption{Graph of the error rate function $n\mapsto -\log(e_n(0))/n$, together with lower and upper bounds.
Starting from below we have the Chernoff bound (the constant),
the lower bound (\ref{eq:upbound}), the exact error rate (the oscillating line), and
the upper bound (\ref{eq:lobound});
the two cases considered are
(a) for $p=0.001$ and $q=0.5$, and
(b) for $p=10^{-10}$ and $q=0.5$.
\label{fig2}
}
\end{center}
\end{figure}

\end{document}